\documentclass[a4paper,USenglish,thm-restate,cleveref]{lipics-v2021}
\pdfoutput=1
\nolinenumbers

\usepackage{complexity}
\usepackage{mathdots}  
\usepackage{csquotes}

\usepackage{xcolor} 

\usepackage[mode=buildnew,obeyclassoptions=true]{standalone}

\usepackage{tikz}
\usetikzlibrary{arrows,automata,calc,positioning}

\makeatletter\let\c@author\relax\makeatother
\usepackage[
  backend=biber,bibencoding=utf8,sortcites=true,maxbibnames=99
]{biblatex}
\AtEveryBibitem{%
  \clearfield{isbn}
  \clearfield{issn}
  \clearfield{note}
  \clearfield{url}
}

\let\H\relax
\let\R\relax
\usepackage{ammacros}

\newclass{\LLT}{LLT}
\newclass{\LTTSS}{LTTSS}
\newclass{\LTT}{LTT}
\newclass{\LT}{LT}
\newclass{\RTIME}{RTIME}
\newclass{\SLT}{SLT}

\newlang{\LLin}{LLin}
\newlang{\Th}{Th}

\title{Sublinear-Time Probabilistic Cellular Automata}

\newcommand{\myname}{Augusto Modanese}
\newcommand{\mynameabbr}{A.~Modanese}
\newcommand{\myaffil}{Aalto University, Espoo, Finland}

\author{\myname}{\myaffil}{augusto.modanese@aalto.fi}{}{}
\authorrunning{\mynameabbr}
\Copyright{\myname}
\ccsdesc{Theory of computation~Formal languages and automata theory}
\keywords{Cellular automata, local computation, probabilistic models, subregular
language classes}
\acknowledgements{I would like to thank Thomas Worsch for the helpful
discussions and feedback as well as the anonymous reviewers for their insightful
comments.}

\begin{document}

\maketitle

\begin{abstract}
  We propose and investigate a probabilistic model of sublinear-time
  one-dimensional cellular automata.
  In particular, we modify the model of ACA (which are cellular automata that
  accept if and only if all cells simultaneously accept) so that every cell
  changes its state not only dependent on the states it sees in its neighborhood
  but also on an unbiased coin toss of its own.
  The resulting model is dubbed \emph{probabilistic ACA} (PACA).
  We consider one- and two-sided error versions of the model (in the same spirit
  as the classes $\RP$ and $\BPP$) and establish a separation between the
  classes of languages they can recognize all the way up to $o(\sqrt{n})$ time.
  As a consequence, we have a $\Omega(\sqrt{n})$ lower bound for derandomizing
  constant-time two-sided error PACAs (using deterministic ACAs).
  We also prove that derandomization of $T(n)$-time PACAs (to polynomial-time
  deterministic cellular automata) for various regimes of $T(n) = \omega(\log
  n)$ implies non-trivial derandomization results for the class $\RP$ (e.g., $\P
  = \RP$).
  The main contribution is an almost full characterization of the constant-time
  PACA classes:
  For one-sided error, the class equals that of the deterministic model; that
  is, constant-time one-sided error PACAs can be fully derandomized with only a
  constant multiplicative overhead in time complexity.
  As for two-sided error, we identify a natural class we call the \emph{linearly
  testable languages} ($\LLT$) and prove that the languages decidable by
  constant-time two-sided error PACAs are \enquote{sandwiched} in-between the
  closure of $\LLT$ under union and intersection and the class of locally
  threshold testable languages ($\LTT$).
\end{abstract}



\section{Introduction}

Cellular automata (CAs) have been extensively studied as a natural model of
distributed computation.
A one-dimensional CA is composed of a row of fairly limited computational
agents---the \emph{cells}---which, by interacting with their immediate
neighbors, realize a global behavior and work towards a common goal.

As every model of computation, CAs have been widely studied as language
acceptors \cite{kutrib09_cellular_ecss, terrier12_language_hnc}.
These efforts apparently were almost exclusively devoted to the linear- or
real-time case---to the detriment of the \emph{sublinear-time} one
\cite{modanese21_sublinear-time_ijfcs}.
This is unfortunate since, as it was recently shown in
\cite{modanese21_lower_csr}, the study of sublinear-time CA variants might help
better direct efforts in resolving outstanding problems in computational
complexity theory.

In this work, we consider a \emph{probabilistic} sublinear-time CA model.
Our main goal is to analyze to what extent---if at all---the addition of
randomness to the model is able to make up for its inherent limitations.
(For instance, sublinear-time CA models are usually restricted to a local view
of their input \cite{modanese21_sublinear-time_ijfcs} and are also unable to
cope with long unary subwords \cite{modanese21_lower_csr}.)

\subsection{The Model}
\label{sec_the_model}

We consider only bounded one-dimensional cellular automata.

\begin{definition}[Cellular automaton]%
  \label{def_ca}%
  A \emph{cellular automaton} (CA) is a triple $C = (Q,\$,\delta)$ where
  $Q$ is the finite set of \emph{states},
  $\$ \notin Q$ is the \emph{boundary symbol},
  and $\delta\colon Q_\$ \times Q \times Q_\$ \to Q$ is the \emph{local
    transition function}, where $Q_\$ = Q \cup \{ \$ \}$.
  The elements in the domain of $\delta$ are the possible \emph{local
    configurations} of the cells of $C$.
  For a fixed width $n \in \N_+$, the \emph{global configurations} of $C$ are
  the elements of $Q^n$.
  The cells $0$ and $n-1$ are the \emph{border cells} of $C$.
  The \emph{global transition function} $\Delta\colon Q^n \to Q^n$ is obtained
  by simultaneous application of $\delta$ everywhere; that is, if $s \in Q^n$
  is the current global configuration of $C$, then
  \[
    \Delta(s) = \delta(\$,s_0,s_1)
      \, \delta(s_0,s_1,s_2)
      \, \cdots
      \, \delta(s_{n-2},s_{n-1},\$).
  \]
  For $t \in \N_0$, $\Delta^t$ denotes the $t$-th iterate of $\Delta$.
  For an initial configuration $s \in Q^n$, the sequence $s = \Delta^0(s),
  \Delta(s), \Delta^2(s), \dots$ is the \emph{orbit} of $C$ (for $s$).
  Writing the orbit of $C$ line for line yields its \emph{space-time diagram}.
\end{definition}

One key theme connecting CAs and models of physics is \emph{causality}:
If two cells $i$ and $j$ are $t$ cells away from each other, then $j$ requires
at least $t$ steps to receive any information from $i$.
In the \emph{sublinear-time} case, this means every cell only gets to see a very
small section of the input.
In some sense this is reminiscent of \emph{locality} in circuits (e.g.,
\cite{yao89_circuits_stoc}), though locality in the CA model carries a more
literal meaning since it is connected to the notion of space (whereas in
circuits there is no equivalent notion).
One should keep this limitation (of every cell only seeing a portion of the
input) in mind as it is central to several of our arguments.

The usual acceptance condition for CA-based language recognizers is that of a
distinguished cell (usually the leftmost one) entering an accepting state
\cite{kutrib09_cellular_ecss}.
This is unsuitable for sublinear-time computation since then the automaton is
limited to verifying prefixes of a constant length
\cite{modanese21_sublinear-time_ijfcs}.
The most widely studied \cite{modanese21_sublinear-time_ijfcs,
  ibarra85_fast_tcs, sommerhalder83_parallel_ai, kim90_characterization_pd}
acceptance condition for sublinear-time is that of all cells simultaneously
accepting, yielding the model of \emph{ACA} (where the first \enquote{A} in the
acronym indicates that \emph{all} cells must accept).

\begin{definition}[DACA]%
  \label{def_aca}%
  A \emph{deterministic ACA} (DACA) is a CA $C$ with an \emph{input alphabet}
  $\Sigma \subseteq Q$ as well as a subset $A \subseteq Q$ of \emph{accepting
  states}.
  We say $C$ \emph{accepts} an input $x \in \Sigma^+$ if there is $t \in \N_0$
  such that $\Delta^t(x) \in A^n$, and we denote the set of all such $x$ by
  $L(C)$.
  In addition, $C$ is said to have \emph{time complexity} (bounded by) $T\colon
  \N_+ \to \N_0$ if, for every $x \in L(C) \cap \Sigma^n$, there is $t <
  T(\abs{x})$ such that $\Delta^t(x) \in A^n$.
\end{definition}

We propose a probabilistic version of the ACA model inspired by the stochastic
automata of \cite{arrighi13_stochastic_fi} and the definition of probabilistic
Turing machines (see, e.g., \cite{arora09_computational_book}).
In the model of \emph{probabilistic ACA} (PACA), at every step, each cell tosses
a fair coin $c \in \binalph$ and then changes its state based on the outcome of
$c$.
There is a nice interplay between this form of randomness access and the overall
theme of \emph{locality} in CAs:
Random events pertaining to a cell $i$ depend exclusively on what occurs in the
vicinity of $i$.
Furthermore, events corresponding to distinct cells $i$ and $j$ can only be
dependent if $i$ and $j$ are near each other; otherwise, they are necessarily
independent (see \cref{lem_independence}).

We consider both one- and two-sided error versions of the model as natural
counterparts of $\RP$ and $\BPP$ machines, respectively.
Although PACAs are a conceptually simple extension of ACAs, the definition
requires certain care, in particular regarding
the model's time complexity.
To see why, recall that, in deterministic ACA (DACA), time complexity of an
automaton $C$ is defined as the upper bound on the number of steps that $C$
takes to accept an input in its language $L(C)$.
In contrast, in a PACA $C'$ there may be multiple computational branches
(depending on the cells' coin tosses) for the same input $x \in L(C')$, and it
may be the case that there is no upper bound on the number of steps for a branch
starting at $x$ to reach an accepting configuration.
In non-distributed models such as Turing machines, these pathological cases can
be dealt with by counting the number of steps computed and stopping if this
exceeds a certain bound.
In PACA, that would either require an extrinsic global agent that informs the
cells when this is the case (which is undesirable since we would like a strictly
distributed model) or it would need to be handled by the cells themselves, which
is impossible in sublinear-time in general (since the cells cannot directly
determine the input length).
We refer to \cref{sec_fundamentals} for the formal definitions and further
discussion.

Finally, we should also mention our model is more restricted than a
\emph{stochastic CA},\footnote{%
  Unfortunately, the literature uses the terms stochastic and probabilistic CA
  interchangeably.
  We deem \enquote{probabilistic} more suitable since it is intended as a CA
  version of a probabilistic Turing machine.}
which is a CA in which the next state of a cell is chosen according to an
arbitrary distribution that depends on the cell's local configuration.
For a survey on stochastic CAs, we refer to \cite{mairesse14_around_tcs}.

\subsection{Results}

\subparagraph{Inclusion relations.}
As can be expected, two-sided error PACAs are more powerful than their one-sided
error counterparts.
Say a DACA $C$ is \emph{equivalent} to a PACA $C'$ if they accept the same
language (i.e., $L(C) = L(C')$).

\begin{theorem}[restate=restatethmOneVsTwosided,name=]%
  \label{thm_one_vs_twosided}%
  The following hold:
  \begin{enumerate}
    \item If $C$ is a one-sided error PACA with time complexity $T$, then there
    is an equivalent two-sided error PACA $C'$ with time complexity $O(T)$.
    \item There is a language $L$ recognizable by constant-time two-sided error
    PACA but not by any $o(\sqrt{n})$-time one-sided error PACA.
  \end{enumerate}
\end{theorem}

We stress the first item does not follow immediately from the definitions since
it requires error reduction by a constant factor, which requires a non-trivial
construction.
It remains open whether in the second item we can improve the separation from
$o(\sqrt{n})$ to $o(n)$ time.
Nevertheless, as it stands the result already implies a lower bound of
$\Omega(\sqrt{n})$ time for the derandomization (to a DACA) of constant-time
two-sided error PACA.

Another result  we show is how time-efficient derandomization of PACA classes
imply derandomization results for $\RP$ (with a trade-off between the PACA time
complexity and the efficiency of the derandomization).

\begin{theorem}[restate=restatethmPACAEqACA,name=]%
  \label{thm_PACA_eq_ACA}%
  Let $d \ge 1$. The following hold:
  \begin{itemize}
    \item If there is $\eps>0$ such that every $n^\eps$-time (one- or two-sided
    error) PACA can be converted into an equivalent $n^d$-time deterministic
    CA, then $\P = \RP$.
    \item If every $\polylog(n)$-time PACA can be converted into an equivalent
    $n^d$-time deterministic CA, then, for every $\eps>0$, $\RP \subseteq
    \TIME[2^{n^\eps}]$.
    \item If there is $b > 2$ so that any $(\log n)^b$-time PACA can be
    converted into an equivalent $n^d$-time deterministic CA, then, for every $a
    \ge 1$ and $c > a/(b-1)$, $\RTIME[n^a] \subseteq \TIME[2^{O(n^c)}]$.
  \end{itemize}
\end{theorem}

We deliberately write \enquote{deterministic CA} instead of \enquote{DACA}
since, for $T(n) = \Omega(n)$, a $T$-time DACA is equivalent to an $O(T)$-time
deterministic CA with the usual acceptance condition
\cite{modanese21_sublinear-time_ijfcs}.

\subparagraph{Characterization of constant time.}
As a first step we analyze and almost completely characterize constant-time
PACA.
Indeed, the constant-time case is already very rich and worth considering in and
of itself.
This may not come as a surprise since other local computational models (e.g.,
local graph algorithms \cite{suomela13_survey_acmcs}) also exhibit behavior in
the constant-time case that is far from trivial.

In \cref{appx_example_daca_vs_paca} we give an example of a one-sided error PACA
that recognizes a language $L$ strictly faster than any DACA for $L$.
Nonetheless, as we prove, one-sided error PACA can be derandomized with only a
constant multiplicative overhead in time complexity.

\begin{theorem}[restate=restatethmCharOnesided,name=]%
  \label{thm_char_onesided_PACA}%
  For any constant-time one-sided error PACA $C$, there is a constant-time DACA
  $C'$ such that $L(C) = L(C')$.
\end{theorem}

In turn, the class of languages accepted by constant-time two-sided error PACA
can be considerably narrowed down in terms of a novel subregular class $\LLT$,
dubbed the \emph{locally linearly testable} languages.
Below, $\LLT_{\cup\cap}$ is the closure over $\LLT$ under union and intersection
and $\LTT$ its Boolean closure (i.e., its closure under union, intersection, and
complement).

\begin{theorem}[restate=restatethmCharTwosided,name=]%
  \label{thm_char_twosided_PACA}%
  The class of languages that can be accepted by a constant-time two-sided error
  PACA contains $\LLT_{\cup\cap}$ and is strictly contained in $\LTT$.
\end{theorem}

It is known that the constant-time class of DACA equals the closure under union
$\SLT_\cup$ of the strictly local languages $\SLT$
\cite{sommerhalder83_parallel_ai}.
(We refer to \cref{sec_char_twosided} for the definitions.)
Since $\SLT_\cup \subsetneq \LLT_\cup$ is a proper inclusion, this gives a
separation of the deterministic and probabilistic classes in the case of
two-sided error and starkly contrasts with \cref{thm_char_onesided_PACA}.

\subparagraph{The class $\LLT$.}

The languages in $\LLT$ are defined based on sets of allowed prefixes and
suffixes (as, e.g., the languages in $\SLT$) together with a \emph{linear
  threshold} condition (hence their name):
For the infixes $m$ of a fixed length $\ell \in \N_+$ there are coefficients
$\alpha(m) \in \R_0^+$ as well as a threshold $\theta \ge 0$ such that every
word $w$ in the language satisfies the following:
\[
  \sum_{m \in \Sigma^\ell} \alpha(m) \cdot \abs{w}_m \le \theta,
\]
where $\abs{w}_m$ is the number of occurrences of $m$ in $w$.

In \cref{sec_char_twosided} we show $\LLT$ lies in-between $\SLT_\cup$ and
the class of locally threshold testable languages $\LTT$.
In this regard $\LLT$ is similar to the class $\LT$ of locally testable
languages; however, as we can also show, both $\LLT$ and $\LLT_\cup$ are
incomparable to $\LT$.
The relation between $\LLT_{\cup\cap}$ and $\LT$ is left as a topic for future
work.

As the classes $\SLT$, $\LT$, and $\LTT$ (see, e.g.,
\cite{beauquier91_languages_tcs}), $\LLT$ may also be characterized in terms of
\emph{scanners}, that is, memoryless devices that process their input by passing
a sliding window of $\ell$ symbols over it.
Namely, the class $\LLT$ corresponds to the languages that can be recognized by
scanners possessing a \emph{single counter} $c$ with maximum value $\theta$; the
counter $c$ is incremented by $\alpha(m)$ for every infix $m \in \Sigma^\ell$
read, and the scanner accepts if and only if $c \le \theta$ holds at the end of
the input (and the prefix and suffix of the input are also allowed).

A related restriction of the $\LTT$ languages that we should mention is that of
the locally threshold testable languages in the strict sense ($\LTTSS$)
\cite{ruiz98_locally_icgi, garcia03_threshold_tcm}.
The key difference between these languages and our class $\LLT$ is that, in the
former, one sets a threshold condition for each infix separately (which
corresponds to using multiple counters in their characterization in terms of
scanners).
In turn, in $\LLT$ there is a \emph{single} threshold condition (i.e., the
inequality above) and in which different infixes may have distinct weights
(i.e., the coefficients $\alpha(m)$).
For instance, this allows counting distinct infixes $m_1$ and $m_2$ towards the
same threshold $t$, which is not possible in the $\LTTSS$ languages (as there
each infix is considered separately).

\subsection{Organization}
The rest of the paper is organized as follows:
\cref{sec_preliminaries} introduces basic concepts and notation.
Following that, in \cref{sec_fundamentals} we define the PACA model and prove
standard error reduction results as well as \cref{thm_one_vs_twosided}.
In \cref{sec_constant_time} we focus on the constant-time case and prove
\cref{thm_char_onesided_PACA,thm_char_twosided_PACA}.
Finally, in \cref{sec_sublinear_time} we address the general sublinear-time
case and prove \cref{thm_PACA_eq_ACA}.
We conclude with \cref{sec_further_directions} by mentioning a few further
research directions.


\section{Preliminaries}
\label{sec_preliminaries}

It is assumed the reader is familiar with the theory of cellular automata as
well as with basic notions of computational complexity theory (see, e.g., the
standard references \cite{goldreich08_computational_book,
arora09_computational_book, delorme99_cellular_book}).

All logarithms are to the base $2$.
The set of integers is denoted by $\Z$, that of non-negative integers by $\N_0$,
and that of positive integers by $\N_+$.
For a set $S$ and $n,m \in \N_+$, $S^{n \times m}$ is the set of $n$-row,
$m$-column matrices over $S$.
For $n \in \N_+$, $[n] = \{ i \in \N_0 \mid i < n \}$ is the set of the first
$n$ non-negative integers.
Also, for $a,b \in \Z$, by $[a,b] = \{ i \in \Z \mid a \le i \le b \}$ we always
refer to an interval containing only integers.

Symbols in words are indexed starting with zero.
The $i$-th symbol of a word $w$ is denoted by $w_i$.
For an alphabet $\Sigma$ and $n \in \N_0$, $\Sigma^{\le n}$ contains the words
$w \in \Sigma^\ast$ for which $\abs{w} \le n$.
For an infix $m \in \Sigma^{\le \abs{w}}$ of $w$, $\abs{w}_m$ is the number of
occurrences of $m$ in $w$.
Without restriction, the empty word is not an element of any language that we
consider.
(This is needed for definitional reasons; see \cref{def_ca,def_aca} below.)

We write $U_n$ (resp., $U_{n \times m}$) for a random variable distributed
uniformly over $\binalph^n$ (resp., $\binalph^{n \times m}$).
We also need the following variant of the Chernoff bound (see, e.g.,
\cite{vadhan12_pseudorandomness_book}):

\begin{theorem}[Chernoff bound]%
  \label{thm_chernoff}%
  Let $X_1,\dots,X_n$ be independently and identically distributed Bernoulli
  variables and $\mu = \Exp[X_i]$.
  There is a constant $c > 0$ such that the following holds for every
  $\eps = \eps(n) > 0$:
  \[
    \Pr\left[ \abs*{\frac{\sum_i X_i}{n} - \mu} > \eps \right]
    < 2^{-cn\eps^2}.
  \]
\end{theorem}

Many of our low-level arguments make use of the notion of a
\emph{lightcone}.\footnote{%
  Some sources distinguish between \emph{future} and \emph{past} lightcones.
  Here we shall only need past lightcones.
}
For a set $S$ and non-negative integers $n \le m$, a lightcone $L =
(\ell_{i,j})$ of \emph{radius} $m$ and \emph{height} $n$ over $S$ is a
trapezoidal (when $n<m$) or triangular (when $n=m$) array of elements
$\ell_{i,j} \in S$, where $i \in [0,n]$ and $j \in [-m,m]$:
\[
  \begin{matrix}
    \ell_{0,-m} & \ell_{0,-m+1} & \cdots & \cdots & \cdots
    & \ell_{0,0} & \cdots & \cdots & \cdots & \ell_{0,m-1} & \ell_{0,m} \\
    & \ell_{1,-m+1} & \cdots & \cdots & \cdots
    & \ell_{1,0}
    & \cdots & \cdots & \cdots & \ell_{1,m-1} \\
    & & \ddots & & & \vdots & & & \iddots \\
    & & & \ell_{n,-m+n} & \cdots & \ell_{n,0} & \cdots & \ell_{n,m-n}
  \end{matrix}
\]
The element $\ell_{0,0}$ is the \emph{center} of the lightcone.
The \emph{layers} of $L$ are indexed by $i$, where the $i$-th layer contains
$2(m-i)+1$ elements.
Hence, the top layer contains $2m+1$ elements and the bottom one $2(m-n)+1$; in
particular, the bottom layer is a single element if and only if $n=m$.
There are $\sum_{i=0}^n (2(m-i)+1) = (n+1)(2m-n+1)$ elements in a lightcone in
total.

\begin{definition}[Neighborhoods and lightcones]
  Let $C$ be a CA and $n \in \N_+$.
  For $i \in [n]$ and $r \in \N_0$, the interval $[i-r,i+r] \cap [n]$ forms the
  \emph{$r$-neighborhood} of $i$.
  For $t \in \N_0$, the \emph{$t$-lightcone} of $i$ is the lightcone of radius
  and height $t$ centered at $i$ in the $0$-th row (i.e., the initial
  configuration) of the space-time diagram of $C$.\footnote{%
    If the lightcone's dimensions overstep the boundaries of the space-time
    diagram (i.e., $i$ is too close to either of the borders of $C$ (e.g.,
    $i<t$)), then some cells in the $t$-lightcone will have undefined states.
    In this case, we set the undefined states to $\$$, which ensures consistency
    with $\delta$.}
\end{definition}

\begin{figure}
  \centering
  \begin{tikzpicture}[font=\large]
  \draw[fill=red] (1,3) rectangle (6,4);

  \draw[fill=blue!50!] (8,3) rectangle (15,4);
  \draw[fill=blue!50!] (9,2) rectangle (14,3);
  \draw[fill=blue!50!] (10,1) rectangle (13,2);
  \draw[fill=blue!50!] (11,0) rectangle (12,1);

  \draw (0,-1) grid (16,4);

  \node at (-0.75,3.5) {$s$};
  \node at (-0.75,2.5) {$\Delta(s)$};
  \node at (-0.75,1.5) {$\Delta^2(s)$};
  \node at (-0.75,0.5) {$\Delta^3(s)$};
  \node at (-0.75,-0.5) {$\vdots$};

  \foreach \i in {0,...,15}
  \node at (\i.5,4.5) {$\i$};
\end{tikzpicture}
  \caption{Space-time diagram of a CA with $16$ cells for an initial
  configuration $s$.
  (States have been omitted for simplicity.)
  The cells marked in red form the $2$-neighborhood of cell $3$, the ones in
  blue the $3$-lightcone of cell number $11$.}
  \label{fig_def_ca}
\end{figure}
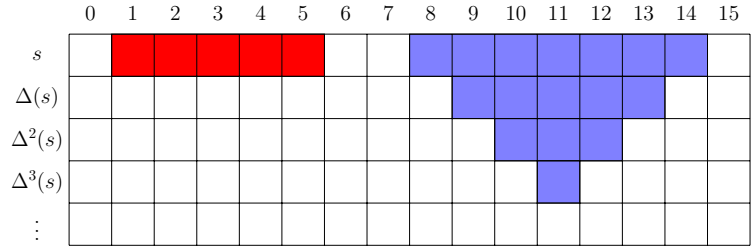


\section{Fundamentals}
\label{sec_fundamentals}

In this section, we introduce the definition of PACA.
Following that, we prove basic error reduction results and conclude with the
proof of \cref{thm_one_vs_twosided}.

As customary for randomized models of computation, one may consider both online
and offline views of our model.
Since it gives a more natural presentation, in the definition below we first
assume an online perspective and then address the definitional issue mentioned
in the introduction.
In the last part, we switch to an offline view that we will use for the rest of
the paper; this is more comfortable to work with since we can then refer to the
cells' coin tosses explicitly.

\begin{definition}[PACA]%
  \label{def_PACA}%
  Let $Q$ be a finite set of states and $\Sigma \subseteq Q$ an alphabet.
  A \emph{probabilistic ACA} (PACA) $C$ is a CA with two local transition
  functions $\delta_0, \delta_1\colon Q^3 \to Q$.
  At each step of $C$, each cell tosses a fair coin $c \in \binalph$ and updates
  its state according to $\delta_c$; that is, if the current configuration of
  $C$ is $s \in Q^n$ and the cells obtain coin tosses $r = r_0 \cdots r_{n-1}
  \in \binalph^n$ (where $r_i$ is the coin toss of the $i$-th cell), then the
  next configuration of $C$ is
  \[
    \Delta_r(s) = \delta_{r_0}(\$,s_0,s_1) \, \delta_{r_1}(s_0,s_1,s_2) \,
      \cdots \, \delta_{r_{n-1}}(s_{n-2},s_{n-1},\$).
  \]
  Seeing this process as a Markov chain $M$ over $Q^n$, we recast the global
  transition function $\Delta = \Delta_{U_n}$ as a family of random variables
  $(\Delta(s))_{s \in Q^n}$ parameterized by the current configuration $s$ of
  $C$, where $\Delta(s)$ is sampled by starting in state $s$ and performing a
  single transition on $M$ (having drawn the cells' coin tosses according to
  $U_n$).
  Similarly, for $t \in \N_0$, $\Delta^t(s)$ is sampled by starting in $s$ and
  performing $t$ transitions on $M$.

  A \emph{computation} of $C$ for an input $x \in \Sigma^n$ is a path in $M$
  starting at $x$.
  The computation is \emph{accepting} if the path visits $A^n$ at least once.
  In order to be able to quantify the probability of a PACA accepting an input,
  we additionally require for every PACA $C$ that there is a function $T\colon
  \N_+ \to \N_0$ such that, for any input $x \in \Sigma^n$, every accepting
  computation for $x$ visits $A^n$ for the first time in strictly less than
  $T(n)$ steps; that is, if there is $t \in \N_0$ with $\Delta^t(x) \in A^n$,
  then $\Delta^{t_1}(x) \in A^n$ for some $t_1 < T(n)$.
  (Hence, every accepting computation for $x$ has an initial segment with
  endpoint in $A^n$ and whose length is strictly less than $T(n)$.)
  If this is the case for any such $T$, then we say $C$ has \emph{time
    complexity} (bounded by) $T$.

  With this condition in place, we may now equivalently replace the coin tosses
  of $C$ with a matrix $R \in \binalph^{T(n) \times n}$ of bits with rows
  $R_0,\dots,R_{T(n)-1}$ and such that $R_j(i)$ corresponds to the coin toss of
  the $i$-th cell in step $j$.
  (If $C$ accepts in step $t$, then the coin tosses in rows $t,\dots,T(n)-1$
  are ignored.)
  We refer to $R$ as a \emph{random input} to $C$.\footnote{%
The number of rows of $R$ is dependent on the choice of $T$.
This is not an issue here since any superficial rows are ignored by $C$; that
is, without restriction we may take $T$ to be such that every value $T(n)$ is
minimal and set the number of rows of $R$ to $T(n)$.
The motivation for letting $R$ be larger is that, when simulating a PACA, it may
be the case that it is more convenient (or even possible) to compute only an
upper bound $T'(n) \ge T(n)$ instead of the actual minimal value $T(n)$.
  }
  Blurring the distinction between the two perspectives (i.e., online and
  offline randomness), we write $C(x,R) = 1$ if $C$ accepts $x$ when its coin
  tosses are set according to $R$, or $C(x,R) = 0$ otherwise.
\end{definition}

As another remark, notice that in \cref{def_PACA} we opt for using binary coin
tosses along with only two local transition functions.
Nonetheless, this is sufficient to realize a set of $2^k$ local transition
functions $\delta_0,\dots,\delta_{2^k-1}$ for constant $k$ with a multiplicative
overhead of $k$.
(Namely, by having each cell collect $k$ coins in $k$ steps, interpret these as
the binary representation of $i \in [2^k]$, and then change its state according
to $\delta_i$.)

\cref{def_PACA} states the acceptance condition for a \emph{single} computation
(i.e., one fixed choice of a random input); however, we must still define 
acceptance based on \emph{all} computations (i.e., for random inputs picked
according to a uniform distribution).
The two most natural candidates are the analogues of the well-studied classes
$\RP$ and $\BPP$, which we define next.

\begin{definition}[$p$-error PACA]%
  \label{def_pPACA}%
  Let $L \subseteq \Sigma^\ast$ and $p \in [0,1)$.
  A \emph{one-sided $p$-error PACA for $L$} is a PACA $C$ with time complexity
  $T=T(n)$ such that, for every $x \in \Sigma^n$,
  \begin{align*}
    x \in L &\iff \Pr[C(x, U_{T \times n}) = 1] \ge 1-p
    && \text{and} & 
    x \notin L &\iff \Pr[C(x, U_{T \times n}) = 1] = 0.
  \end{align*}
  If $p = 1/2$, then we simply say $C$ is a \emph{one-sided error PACA}.
  Similarly, for $p < 1/2$, a \emph{two-sided $p$-error PACA for $L$} is a PACA
  $C$ with time complexity $T=T(n)$ for which
  \begin{align*}
    x \in L &\iff \Pr[C(x, U_{T \times n}) = 1] \ge 1-p
    && \text{and} & 
    x \notin L &\iff \Pr[C(x, U_{T \times n}) = 1] \le p
  \end{align*}
  hold for every $x \in \Sigma^\ast$.
  If $p = 1/3$, then we simply say $C$ is a \emph{two-sided error PACA}.
  In both cases, we write $L(C) = L$ and say $C$ \emph{accepts} $L$.
\end{definition}

Note that, to each $0$-error PACA $C$, one can obtain an equivalent DACA $C'$
with the same time complexity by setting the local transition function to
$\delta_0$.
In the rest of the paper, if it is not specified which of the two variants above
(i.e., one- or two-sided error) is meant, then we mean both variants
collectively.
See \cref{appx_example_daca_vs_paca} for an example of PACAs being more
efficient than DACAs.

From the perspective of complexity theory, it is interesting to compare the PACA
model with probabilistic circuits.
It is known that every $T(n)$-time DACA can be simulated by an $\L$-uniform AC
circuit (i.e., a Boolean circuit with gates of unbounded fan-in) having
$\poly(n)$ size and $O(\max\{1, T(n)/\log n\})$ depth
\cite{modanese21_sublinear-time_ijfcs}.
Using the same approach as in \cite{modanese21_sublinear-time_ijfcs}, we note
the same holds for PACAs if we use probabilistic AC circuits instead.
The proof in \cite{modanese21_sublinear-time_ijfcs} bases on descriptive
complexity theory, the central observation being that the state of a cell $i$
after $\log n$ steps given its ($\log n$)-neighborhood is a predicate that is
computable in logarithmic space.
Hence, for the PACA case we need only factor in the auxiliary random input into
this predicate.

A key property that PACAs have but probabilistic circuits do not, however, is
\emph{distance} between computational units.
(Indeed, in circuits, there is no such thing as the \enquote{length} of a wire.)
One consequence of this is the following simple fact.

\begin{lemma}[Independence of local events]%
  \label{lem_independence}%
  Let $C$ be a one- or two-sided error PACA, let $x \in \Sigma^n$ be an input to
  $C$, and let $T \in \N_+$.
  In addition, let $i,j \in [n]$ be such that $\abs{i-j} > 2(T-1)$ and $E_i$
  (resp., $E_j$) be an event described exclusively by the states of the $i$-th
  (resp., $j$-th) cell of $C$ in the time steps $0,\dots,T-1$ (e.g., the $i$-th
  cell accepts in some step $t$ where $t < T$).
  Then $E_i$ and $E_j$ are independent.
\end{lemma}

\begin{proof}
  For any random input $R$, the states of $k \in \{i,j\}$ in the time steps
  between $0$ and $T-1$ is uniquely determined by the values of
  $R(t,k-T+t+1),\dots,R(t,k+T-t-1)$ for $t \in [T]$.
  Without loss of generality, suppose $i \le j$.
  Since $i+T-1 < j-T+1$, $E_i$ and $E_j$ are conditioned on disjoint sets
  of values of $R$, thus implying independence.
\end{proof}

Note the proof still holds in case $T = 1$, in which case the events $E_i$ and
$E_j$ occur with probability either $0$ or $1$, thus also (trivially) implying
independence.

\subsection{Robustness of the Definition}
\label{sec_error_reduction}

We now prove that the definition of PACA is robust with respect to the choice of
$p = 1/2$ (resp., $p = 1/3$) for the error of one-sided (resp., two-sided) error
PACA.

\subsubsection{One-Sided Error}

For one-sided error, we can reduce the error $p$ to any desired constant value
$p'$.

\begin{proposition}[restate=restatepropOnesidedErrorReduction,name=]%
  \label{prop_onesided_error_reduction}%
  Let $p, p' \in (0,1)$ be constant and $p' < p$.
  For every one-sided $p$-error PACA $C$, there is a one-sided $p'$-error PACA
  $C'$ such that $L(C) = L(C')$.
  Furthermore, if $C$ has time complexity $T(n)$, then $C'$ has time complexity
  $O(T(n))$.
\end{proposition}

It follows that the definition of PACA is robust under the choice of $p$ (as
long as it is constant) and regardless of the time complexity (up to constant
multiplicative factors).

The proof is essentially a generalization of the idea used in
\cite{modanese21_sublinear-time_ijfcs} to show that the sublinear-time DACA
classes are closed under union.
Namely, $C'$ simulates several copies $C_0,\dots,C_{m-1}$ of $C$ in parallel and
accepting if and only if at least one $C_i$ accepts.
This idea is particularly elegant because $m$ can be chosen to be constant and
we update the $C_i$ in a round-robin fashion (i.e., first $C_0$, then $C_1$,
$C_2$, etc., and finally $C_0$ again after $C_{m-1}$).
The alternative is to simulate each $C_i$ for $T(n)$ steps at a time, which is
not possible in general since we would have to compute $T(n)$ first.
The construction we give avoids this issue entirely.

\begin{proof}
  We construct a PACA $C'$ with the desired properties.
  Let $m = \ceil{\log(1/p' - 1/p)}$.
  Furthermore, let $Q$ be the state set of $C$ and $\Sigma$ its input alphabet.
  We set the state set of $C'$ to $Q^m \times [m] \cup \Sigma$.
  Given an input $x$, every cell of $C'$ initially changes its state from $x(i)$
  to $(x(i),\dots,x(i),0)$.
  The cells of $C'$ simulate $m$ copies of $C$ as follows:
  If the last component of a cell contains the value $j$, then its $j$-th
  component\footnote{%
    In the same manner as we do for the indices of a word, we number the
    components starting with zero.}
  $q_0$ is updated to $\delta(q_{-1},q_0,q_1)$, where $q_{-1}$ and $q_1$ are the
  $j$-th components of the left and right neighbors, respectively (or $\$$ in
  case of a border cell); at the same time, the last component of the cell is
  set to $j+1$ if $j < m$ or $0$ in case $j = m$.
  A cell of $C'$ is accepting if and only if its last component is equal to $j$
  and its $j$-th component is an accepting state of $C$.

  Denote the $i$-th simulated copy of $C$ by $C_i$.
  Clearly, $C'$ accepts in step $mt+i+1$ for $i \in [m]$ if and only if $C_i$
  accepts in step $t$, so we immediately have that $C'$ has $O(T(n))$ time
  complexity.
  For the same reason and since $C'$ never accepts in step $0$, $C'$ does not
  accept any input $x \notin L(C)$.
  As for $x \in L(C)$, note the $m$ copies of $C$ are all simulated using
  independent coin tosses, thus implying
  \[
    \Pr[\text{$C'$ does not accept $x$}]
    = \Pr[\forall i \in [m]: \text{$C_i$ does not accept $x$}]
    < p^m
    \le p'.
  \]
  Hence, $C'$ accepts $x$ with probability at least $1-p'$, as desired.
\end{proof}

\subsubsection{Two-Sided Error}

For two-sided error, we show the same holds for every choice of $p$ for
\emph{constant-time} PACA.
We remark the construction is considerably more complex than in the one-sided
error case.

\begin{proposition}[restate=restatepropTwosidedErrorReduction,name=]%
  \label{prop_twosided_error_reduction}%
  Let $p,p' \in (0,1/2)$ be constant and $p' < p$.
  For every two-sided $p$-error PACA $C$ with constant time complexity $T =
  O(1)$, there is a two-sided $p'$-error PACA $C'$ with time complexity $O(T) =
  O(1)$ and such that $L(C) = L(C')$.
\end{proposition}

To reduce the error, we use the standard method based on the Chernoff bound
(\cref{thm_chernoff}); that is, the PACA $C'$ simulates $m$ independent copies
$C_0,\dots,C_{m-1}$ of $C$ (for an adequate choice of $m$) and then accepts if
and only if the majority of the $C_i$ do.
More precisely, $C'$ loops over every possible majority $\mathcal{M} \subseteq
[m]$ (i.e., every set $\mathcal{M} \subseteq [m]$ with $\abs{\mathcal{M}} \ge
m/2$) over the $C_i$ and checks whether $C_i$ accepts for every $i \in
\mathcal{M}$ (thus reducing majority over the $L(C_i)$ to intersection over the
$L(C_j)$ where $j \in \mathcal{M}$).
In turn, to check whether every $C_j$ accepts for $j \in \mathcal{M}$, $C'$
tries every possible combination of time steps for the $C_j$ to accept and
accepts if such a combination is found. If this entire process fails, then the
majority of the $C_i$ do not accept, and thus $C'$ does not accept as well.
(Obviously, this idea is only feasible if $m$ as well as the time complexities
of the $C_i$ are constant.)

\begin{proof}
  Let $Q$ be the state set of $C$ and $\Sigma$ its input alphabet.
  We also fix a constant $m$ depending only on $p$ and $p'$ which will be set
  later and let $M = \binom{m}{\ceil{m/2}}$.

  \proofsubparagraph{Construction.}
  The state set of $C'$ is $Q' \cup \Sigma$, where $Q'$ is a set of states
  consisting of the following components:
  \begin{itemize}
    \item an $m$-tuple from $Q^m$ of states of $C$
    \item an $m$-tuple from $[T]^m$ representing an $m$-digit, $T$-ary counter
    $i_0 \cdots i_{m-1}$
    \item a value from $[M]$ representing a counter modulo $M$
    \item an input symbol from $\Sigma$
    \item a $(T \times m)$-matrix of random bits, which are all initially set to
    an undefined value different from $0$ or $1$
  \end{itemize}
  Given an input $x$, in the first step of $C'$ every cell changes its state
  from $x(i)$ to a state in $Q'$ where the $Q^m$ components are all set to
  $x(i)$, the numeric ones (i.e., with values in $[T]^m$ and $[M]$) set to $0$,
  and the value $x(i)$ is stored in the $\Sigma$ component.
  In the first phase of $C'$, which lasts for $mT$ steps, the cells fill their
  $(T \times m)$-matrix with random bits.
  This is the only part of the operation of $C'$ in which its random input is
  used (i.e., in all subsequent steps, $C'$ operates deterministically and the
  outcome of the remaining coin tosses is ignored).

  In this next phase, $C'$ simulates $m$ copies $C_0,\dots,C_{m-1}$ of $C$ in
  its $Q^m$ components (as in the proof of
  \cref{prop_onesided_error_reduction}).
  The random bits for the simulation are taken from the previously filled $(T
  \times m)$-matrix, where the $T$ entries in the $i$-th column are used as the
  coin tosses in the simulation of $C_i$ (with the entry in the respective
  $j$-th row being used in the $j$-th step of the simulation).
  Meanwhile, the $i_0 \cdots i_{m-1}$ counter is taken to represent that the
  simulation of $C_j$ is in its $i_j$-th step.

  The cells update their states as follows:
  At each step, the counter is incremented and the respective simulations are
  updated accordingly; more precisely, if $C_j$ is in step $i_j$ and $i_j$ was
  incremented (as a result of the counter being incremented), then the
  simulation of $C_j$ is advanced by one step; if the value $i_j$ is reset to
  $0$, then the simulation of $C_j$ is restarted by setting the respective state
  to $x(i)$.
  Every time the counter has looped over all possible values, the $[M]$
  component of the cell is incremented (and the process begins anew with the
  counter set to all zeroes).
  When the value of the $[M]$ component is equal to $M-1$ and the counter
  reaches its final value (i.e., $i_j = T-1$ for every $j$), then the cell
  conserves its current state indefinitely.

  We agree upon an enumeration $\mathcal{M}_0,\dots,\mathcal{M}_{M-1}$ of the
  subsets of $[m]$ of size $\ceil{m/2}$ and identify a value of $i$ in the $[M]$
  component with $\mathcal{M}_i$.
  A cell whose $[M]$ component is equal to $i$ is then accepting if and only if,
  for every $j \in \mathcal{M}_i$, its $j$-th component is an accepting state of
  $C$.

  \proofsubparagraph{Correctness.}
  By construction, if $C'$ accepts in a time step where the $T$-ary counters
  have the value $i_0 \cdots i_{m-1}$ and the $[M]$ component the value $j$,
  then this is the case if and only if $C_k$ accepts in step $i_k$ for every $k
  \in \mathcal{M}_j$.
  Hence, $C'$ accepts if and only if at least $\ceil{m/2}$ of the simulated
  copies of $C$ accept (which, by definition, must occur in a time step prior to
  $T$); that is, $C'$ accepts if and only if a majority of the
  $C_0,\dots,C_{m-1}$ accept.

  Finally, we turn to setting the parameter $m$ so that $C'$ only errs with
  probability at most $p'$.
  Let $X_i$ be the random variable that indicates whether $C_i$ accepts
  conditioned on its coin tosses.
  Then the probability that $C'$ errs is upper-bounded by the probability that
  $(\sum_i X_i)/m$ deviates from the mean $\mu = \Pr[C(x,r) = 1] \ge 1-p$ by
  more than $\eps = 1/2 - p$.
  By the Chernoff bound (\cref{thm_chernoff}), this occurs with probability at
  most $2^{-cm\eps^2}$ for some constant $c > 0$, so setting $m$ such that $m
  \ge \log(1/p')/c\eps^2$ completes the proof.
\end{proof}

It remains open whether a similar result holds for general (i.e.,
non-constant-time) two-sided error PACA.
Generalizing our proof of \cref{prop_twosided_error_reduction} would require at
the very least a construction for intersecting non-constant-time PACA languages.
(Note we do show closure under intersection for the \emph{constant-time}
languages later in \cref{prop_closure_union}.)
If such a construction were to be known, then extending the idea above one could
use that the union of constantly many $T$-time PACA languages can be recognized
in $O(T)$ time (as we prove later in \cref{prop_closure_union}) and represent
the majority over $m$ PACA languages $L_0,\dots,L_{m-1}$ as the union over all
possible intersections of $\ceil{m/2}$ many $L_i$.
Note that closure under intersection is open in the deterministic setting (i.e.,
of DACA) as well \cite{modanese21_sublinear-time_ijfcs}.

\subsection{One- vs. Two-Sided Error}

The results of \cref{sec_error_reduction} are also useful in obtaining the
following:

\restatethmOneVsTwosided*

\begin{proof}
  The first item follows from \cref{prop_onesided_error_reduction}:
  Transform $C$ into a one-sided error PACA $C'$ with error at most $1/3$ and
  then notice that $C'$ also qualifies as a two-sided error PACA (as it simply
  never errs on \enquote{no} instances).
  For the second item, consider the language
  \[
    L = \{ x \in \binalph^+ \mid \abs{x}_1 \le 1 \}.
  \]
  We obtain a constant-time two-sided error PACA for $L$ as follows:
  If a cell receives a $0$ as input, then it immediately accepts; otherwise, it
  collects two random bits $r_0$ and $r_1$ in the first two steps and then,
  seeing $r_0 r_1$ as the binary representation of an integer $1 \le t \le 4$,
  it accepts (only) in the subsequent $t$-th step.
  Hence, if the input $x$ is such that $\abs{x}_1 \le 1$, the PACA always
  accepts; conversely, if $\abs{x}_1 \ge 2$, then the PACA only accepts if all
  $1$ cells pick the same value for $t$, which occurs with probability at most
  $1/4$.

  It remains to show $L(C) \neq L$ for any $T$-time one-sided error PACA $C$
  where $T = o(\sqrt{n})$.
  Let $n$ be large enough so that $T = T(n) \le \sqrt{n}/2$.
  Observe that $L \cap \binalph^n = \{ 0^n, x_1, \dots, x_n \}$ where $x_i =
  0^{i-1}10^{n-i}$.
  Let us now assume that $x_i \in L(C)$ holds for every $i$.
  Since $C$ accepts with probability at least $1/2$, by the pigeonhole principle
  there is $R$ such that $C(x_i,R) = 1$ for at least a $1/2$ fraction of the
  $x_i$.
  In addition, by averaging there is a step $t < T$ such that at least a $1/2T$
  fraction of the $x_i$ is accepted by $C$ in step $t$.
  Since there are $n/2T \ge 2T \ge 2t + 2$ such $x_i$, we can find $i,j \in
  [n]$ with $j \ge i + 2t + 1$ and $x_i,x_j \in L(C)$.
  Consider now the input
  \[
    x^\ast = 0^{i-1}10^{j-i-1}10^{n-j},
  \]
  which is not in $L$.
  We argue $C(x^\ast,R) = 1$, thus implying $L(C) \neq L$ and completing the
  proof.

  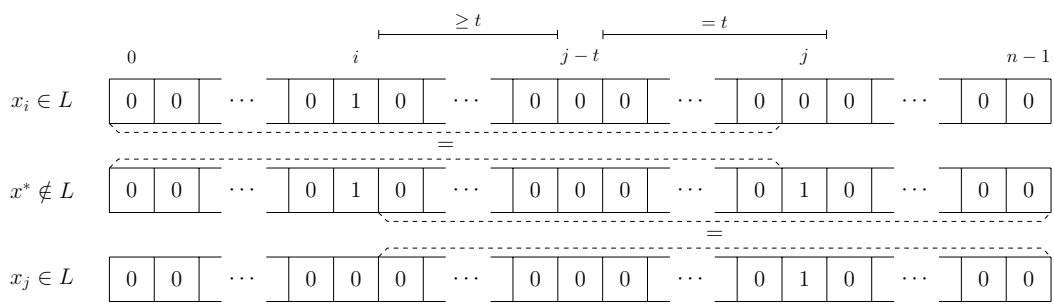
\begin{figure}
    \centering
    \begin{tikzpicture}[font=\large]
  \foreach \j in {0,-2,-4} {
    \draw (0,\j) grid (2.5,\j+1);
    \draw (3.5,\j) grid (7.5,\j+1);
    \draw (8.5,\j) grid (12.5,\j+1);
    \draw (13.5,\j) grid (17.5,\j+1);
    \draw (18.5,\j) grid (21,\j+1);
  }

  \begin{scope}[shift={(0.5,0.5)}]
    \node at (-2,0) {$x_i \in L$};
    \foreach \i in {0,1,4,6,9,10,11,14,15,16,19,20} \node at (\i,0) {$0$};
    \foreach \i in {5} \node at (\i,0) {$1$};
    \foreach \i in {2.5,7.5,12.5,17.5} \node at (\i,0) {$\cdots$};
  \end{scope}

  \begin{scope}[shift={(0.5,-1.5)}]
    \node at (-2,0) {$x^\ast \notin L$};
    \foreach \i in {0,1,4,6,9,10,11,14,16,19,20} \node at (\i,0) {$0$};
    \foreach \i in {5,15} \node at (\i,0) {$1$};
    \foreach \i in {2.5,7.5,12.5,17.5} \node at (\i,0) {$\cdots$};
  \end{scope}

  \begin{scope}[shift={(0.5,-3.5)}]
    \node at (-2,0) {$x_j \in L$};
    \foreach \i in {0,1,4,5,6,9,10,11,14,16,19,20} \node at (\i,0) {$0$};
    \foreach \i in {15} \node at (\i,0) {$1$};
    \foreach \i in {2.5,7.5,12.5,17.5} \node at (\i,0) {$\cdots$};
  \end{scope}

  \begin{scope}[shift={(0.5,1.5)},font=\normalsize]
    \node at (0,0) {$0$};
    \node at (5,0) {$i$};
    \node at (10,0) {$j-t$};
    \node at (15,0) {$j$};
    \node at (20,0) {$n-1$};

    \draw[|-|] (5.5,0.5) -- (9.5,0.5) node[midway,above] {$\ge t$};
    \draw[|-|] (10.5,0.5) -- (15.5,0.5) node[midway,above] {$= t$};
  \end{scope}

  \draw[dashed] (0,0) -- (0.2,-0.2) -- (14.8,-0.2) -- (15,0);
  \draw[dashed] (0,-1) -- (0.2,-0.8) -- (14.8,-0.8) -- (15,-1);
  \draw[dashed] (6,-2) -- (6.2,-2.2) -- (20.8,-2.2) -- (21,-2);
  \draw[dashed] (6,-3) -- (6.2,-2.8) -- (20.8,-2.8) -- (21,-3);

  \node at (7.5,-0.5) {$=$};
  \node at (13.5,-2.5) {$=$};
\end{tikzpicture}
    \caption{Constructing $x^\ast \notin L$ from $x_i,x_j \in L$.
      The numbers above the cells indicate their respective indices.
      Since every $t$-neighborhood of $x^\ast$ appears in either $x_i$ or $x_j$
      and both $x_i$ and $x_j$ are accepted in (exactly) $t$ steps, it follows
      that $C$ accepts $x^\ast$ in $t$ steps.}
    \label{fig_thm_one_vs_twosided}
  \end{figure}

  We can see this by comparing the local views of the \enquote{bad} word
  $x^\ast$ with the \enquote{good} ones $x_i$ and $x_j$ (see
  \cref{fig_thm_one_vs_twosided}):
  Let $k \in [n]$ be any cell of $C$.
  If $k < j-t$, then the $t$-neighborhood of $k$ on input $x^\ast$ is identical
  to that when the input is $x_i$, so $k$ must be accepting in step $t$.
  Similarly, if $k \ge j-t$, then the $t$-neighborhood of $k$ in $x^\ast$ is
  the same as in $x_j$, so $k$ is accepting as well.
  It follows all cells of $C$ are accepting in step $t$ for inputs $x^\ast$ and
  $R$.
\end{proof}


\section{The Constant-Time Case}
\label{sec_constant_time}

In this section we now focus on constant-time PACA.
Our goal will be to characterize the constant-time classes of both one- and
two-sided error PACA (i.e.,
\cref{thm_char_onesided_PACA,thm_char_twosided_PACA}).
First, we introduce the concept of \emph{critical cells}, which is central to
our analysis.

\subsection{Critical Cells}

\begin{definition}[Critical cell]
  Let $C$ be a one- or two-sided error PACA, and let $x \in L(C) \cap \Sigma^n$.
  We say a cell $i \in [n]$ is \emph{critical} for $x$ in step $t \in \N_0$ if
  there are random inputs $R,R' \in \binalph^{t \times n}$ such that $i$ is
  accepting in step $t$ of $C(x,R)$ but not in step $t$ of $C(x,R')$.
\end{definition}

In other words, if $E$ is the event of $i$ being accepting in step $t$ of $C$ on
input $x$, then $0 < \Pr[E] < 1$ (where the probability is taken over the coin
tosses of $C$).
We should stress that whether a cell is critical or not may be \emph{highly
dependent} on $x$ and $t$; for instance, there may be inputs $x_1 \neq x_2$
where the cell $i$ is critical for $x_1$ but not for $x_2$.

As it turns out, the number of critical cells of a constant-time PACA is also
constant.

\begin{lemma}[restate=restatelemCriticalCells, name=]%
  \label{lem_critical_cells}%
  Let $C$ be a $T$-time (one- or two-sided error) PACA for $T \in \N_+$, and let
  $x \in L(C) \cap \Sigma^n$.
  In addition, let $t \in [T]$ be a time step in which $x$ is accepted by $C$
  with non-zero probability.
  Then there are $2^{O(T^2)}$ cells that are critical for $x$ in step $t$.
  It follows there are $T \cdot 2^{O(T^2)} = 2^{O(T^2)}$ critical cells for $x$
  in total (i.e., all such time steps comprised).
\end{lemma}


\begin{proof}
  For $i \in [n]$, let $E_i$ denote the event in which the $i$-th cell accepts
  in step $t$.
  Assume towards a contradiction that there is $x \in L(C)$ in which strictly
  more than $2T \cdot (1+\log T) \cdot 2^{T^2} = 2^{O(T^2)}$ cells are critical
  for $x$ (in step $t$).
  By inspection, this implies there is a set $K \subseteq [n]$ of $\abs{K} \ge
  (1+\log T) \cdot 2^{T^2}$ cells such that every $k \in K$ is critical for $x$
  and, for every distinct $i,j \in K$, $\abs{i-j} \ge 2T$.
  (For instance, in the extreme case where every $0 \le i < 2T \cdot 2^{T^2}$
  is critical, pick $K = \{ 2jT \mid j \in [2^{T^2}] \}$.)
  Since there are at most $T^2$ coin tosses that determine whether a cell $i
  \in K$ accepts or not (i.e., those in the $T$-lightcone of $i$), $\Pr[E_i] <
  1$ if and only if $\Pr[E_i] \le 1 - 2^{-T^2}$.
  By \cref{lem_independence}, the events $E_i$ are all independent, implying
  \[
    \Pr[\text{$C$ accepts $x$ in step $t$}]
    \le \prod_{i \in K} \Pr[E_i]
    \le \left(1 - \frac{1}{2^{T^2}}\right)^{\abs{K}}
    < \frac{1}{e^{1+\log T}}
    < \frac{1}{2T}.
  \]
  Since $t$ was arbitrary, by a union bound it follows that the probability that
  $C$ accepts $x$ in step $t$ is strictly less than $1/2$.
  This contradicts $x \in L(C)$ both when $C$ is a one- and a two-sided error
  PACA.
\end{proof}



\subsection{Characterization of One-Sided Error PACA}
\label{sec_char_onesided}

\restatethmCharOnesided*

\cref{lem_critical_cells} implies that, given any $x \in L(C)$, the decision of
$C$ accepting can be traced back to a set $K$ of critical cells where $\abs{K}$
is constant.
To illustrate the idea, suppose that, if $C$ accepts, then it always does so in
a fixed time step $t < T$; in addition, assume the cells in $K$ are all far
apart (e.g., more than $2(T-1)$ cells away from each other as in
\cref{lem_independence}).

Let us \emph{locally} inspect the space-time diagram of $C$ for $x$, that is, by
looking at the $t$-lightcone of each cell $i$.
Then we notice that, if $i \in K$, there is a choice of random bits in the
$t$-lightcone of $i$ that causes $i$ to accept; conversely, if $i \notin K$,
then \emph{any} setting of random bits results in $i$ accepting.
Consider how this changes when $x \notin L(C)$ while assuming $K$ remains
unchanged.
Every cell $i \notin K$ still behave the same; that is, it accepts regardless
of the random input it sees.
As for the cells in $K$, however, since they are all far apart, it cannot be the
case that we still find random bits for every $i \in K$ that cause $i$ to
accept; otherwise $C$ would accept $x$ with non-zero probability, contradicting
the definition of one-sided error PACA.
Hence, there must be at least some $i \in K$ that \emph{never} accepts.
In summary, (under these assumptions) we can locally distinguish $x \in L(C)$
from $x \notin L(C)$ by looking at the cells of $K$ and checking whether, for
every $i \in K$, there is at least one setting of the random bits in the
$t$-lightcone of $i$ that causes it to accept.

To obtain the proof, we must generalize this idea to handle the case where the
cells in $K$ are not necessarily far from each other---which in particular means
we can no longer assume that the events of them accepting are independent---as
well as of $K$ varying with the input.
Since \cref{lem_critical_cells} only gives an upper bound for critical cells
when the input is in $L(C)$, we must also account for the case where the input
is not in $L(C)$ and $\abs{K}$ exceeds said bound.

\begin{proof}
  Let $T \in \N_+$ be the time complexity of $C$, and let $M$ be the upper bound
  on the number of critical cells (for inputs in $L(C)$) from
  \cref{lem_critical_cells}.
  Without restriction, we may assume $T > 1$.
  We first give the construction for $C'$ and then prove its correctness.

  \proofsubparagraph{Construction.}
  Given an input $x \in \Sigma^n$, the automaton $C'$ operates in two phases.
  In the first one, the cells communicate so that, in the end, each cell is
  aware of the inputs in its $r$-neighborhood, where $r = (2M-1)(T-1)$.
  (Note this is possible because $r$ is constant.)
  The second phase proceeds in $T$ steps, with the cells assuming accepting
  states or not depending on a condition we shall describe next.
  After the second phase is over (and $C'$ has not yet accepted), all cells
  unconditionally enter a non-accepting state and maintain it indefinitely.
  Hence, $C'$ only ever accepts during the second phase.

  We now describe when a cell $i \in [n]$ is accepting in the $t$-th step of the
  second phase, where $t \in [T]$.
  Let $K$ be the set of critical cells of $x$ in step $t$.
  The decision process is as follows:
  \begin{enumerate}
    \item First the cell checks whether there are strictly more than $M$ cells
    in its $r$-neighborhood $N$ that are critical in step $t$ of $C$.
    If $i$ cannot determine this for any cell $j \in N$ (since it did not
    receive the entire $t$-neighborhood of $j$ during the first phase), then $j$
    is simply ignored.
    If this is the case, then $i$ assumes a non-accepting state (in the $t$-th
    step of the second phase).
    \item The cell then determines whether it is critical itself in step $t$ of
    $C$.
    If this is \emph{not} the case, then it becomes accepting if and only if it
    is accepting in step $t$ of $C$ (regardless of the random input).
    \item Otherwise $i$ is critical in step $t$.
    Let $B_i \subseteq [n]$ be the subset of cells that results from the
    following sequence of operations:
    \begin{enumerate}
      \item Initialize $B_i$ to $\{ i \}$.
      \item For every cell $j \in B_i$, add to $B_i$ every $k \in K$ such that
      $\abs{j-k} \le 2(T-1)$.
      \item Repeat step $2$ until a fixpoint is reached.
    \end{enumerate}
    (This necessarily terminates since there are at most $M$ critical cells in
    $N$ and we checked the upper bound of $M$ previously.)
    By choice of $r$, we have that $\abs{i-j} \le 2(M-1)(T-1)$ for every $j \in
    B_i$.
    In particular, we have that the $t$-neighborhood of every $j \in B_i$ is
    completely contained in $N$, which means cell $i$ is capable of determining
    $B_i$.\footnote{%
      Note it is not necessary for $i$ to be aware of the actual numbers of the
      cells in $B_i$; it suffices for it to compute their positions relative to
      itself.
      For example, if $i=5$ and $B_i = \{ 4,5 \}$, then it suffices for $i$ to
      regard $j=4$ as cell $-1$ (relative to itself).
      Hence, by \enquote{determining $B_i$} here we mean that $i$ computes only
      these relative positions (and not the absolute ones, which would be
      impossible to achieve in only constant time).
    }
    The cell $i$ then accepts if and only if there is a setting of random bits
    in the lightcone $L_i$ of radius $r$ and height $t$ centered at $i$ that
    causes every cell in $B_i$ to accept in step $t$ of $C$.
  \end{enumerate}
  This process repeats itself in every step $t$ of the second phase.
  Note that it can be performed by $i$ instantaneously (i.e., without requiring
  any additional time steps of $C'$) since it can be hardcoded into the local
  transition function $\delta$.

  \proofsubparagraph{Correctness.}
  It is evident that $C'$ is a constant-time PACA, so all that remains is to
  verify its correctness.
  To that end, fix an input $x$ and consider the two cases:
  \begin{description}
    \item[$x \in L(C)$.]
    Then there is a random input $R$ such that $C$ accepts $x$ in step $t \in
    [T]$.
    This means that, for every critical cell $i \in K$, if we set the random
    bits in $L_i$ according to $R$, then every cell in $B_i$ accepts in step $t$
    of $C$.
    Likewise, every cell $i \notin K$ is accepting in step $t$ of $C$ by
    definition.
    In both cases we have that $i$ also accepts in the $t$-th step of the second
    phase of $C'$, thus implying $x \in L(C')$.

    \item[$x \notin L(C)$.]
    Then, for every random input $R$ and every step $t \in [T]$, there is at
    least one cell $i \in [n]$ that is not accepting in the $t$-th step of
    $C(x,R)$.
    If $i$ is not critical, then $i$ is also not accepting in the $t$-th step
    of the second phase of $C'$ (regardless of the random input), and thus $C'$
    also does not accept $x$.
    Hence, assume that every such $i$ (i.e., every $i$ such that there is a
    random  input $R$ for which $i$ is not accepting in the $t$-th step of
    $C(x,R)$) is a critical cell.

    Let $J \subseteq [n]$ denote the set of all such cells and, for $i \in J$,
    let $D_i \subseteq J$ be the subset that contains every $j \in J$ such that
    the events of $i$ and $j$ accepting in step $t$ are \emph{not} independent
    (conditioned on the random input to $C$).
    In addition, let $A_i$ denote the event in which every cell of $D_i$ is
    accepting in step $t$ of $C(x,U_{T \times n})$.
    We show the following:
    \begin{claim*}
      There is an $i \in J$ such that $\Pr[A_i] = 0$; that is, for every $R$,
      there is at least one cell in $D_i$ that is not accepting in step $t$ of
      $C(x,R)$.
    \end{claim*}
    This will complete the proof since then $i$ is also not accepting in the
    $t$-th step of the second phase of $C'$ (since any cell in $D_i$ is
    necessarily at most $2(M-1)(T-1)$ cells away from $i$), thus implying $x
    \notin L(C')$.

    To see the claim is true, suppose towards a contradiction that, for every $i
    \in J$, we have $\Pr[A_i] > 0$.
    If there are $i,j \in J$ such that $j \notin D_i$ (and similarly $i \notin
    D_j$), then by definition
    \[
      \Pr[C(x,U_{T \times n}) = 1]
      \ge \Pr[A_i \land A_j]
      = \Pr[A_i] \Pr[A_j]
      > 0,
    \]
    contradicting $x \notin L(C)$.
    Thus, there must be $i \in J$ such that $J = D_i$; however, this then
    implies
    \[
      \Pr[C(x,U_{T \times n}) = 1] = \Pr[A_i] > 0.
    \]
    As this is also a contradiction, the claim (and hence the theorem) follows.
    \qedhere
  \end{description}
\end{proof}

\subsection{Characterization of Two-Sided Error PACA}
\label{sec_char_twosided}

This section in divided into two parts.
In the first, we introduce the class $\LLT$ of locally linearly testable
languages and relate it to other classes of subregular languages.
The second part covers the proof of \cref{thm_char_twosided_PACA} proper.

\subsubsection{Local Languages}
\label{sec_local_languages}

We introduce some notation.
For $\ell \in \N_0$ and a word $w \in \Sigma^\ast$, $p_\ell(w)$ is the prefix of
$w$ of length $\ell$ if $\abs{w} \ge \ell$, or $w$ otherwise; similarly,
$s_\ell(w)$ is the suffix of length $\ell$ if $\abs{w} \ge \ell$, or $w$
otherwise.
The set of infixes of $w$ of length (exactly) $\ell$ is denoted by $I_\ell(w)$.

The subregular language classes from the next definition are due to
\citeauthor{mcnaughton71_counter-free_book}
\cite{mcnaughton71_counter-free_book} and
\citeauthor{beauquier89_factors_icalp}
\cite{beauquier89_factors_icalp}.

\begin{definition}[SLT, LT, LTT]%
  \label{def_SLT_LT_LTT}%
  A language $L \subseteq \Sigma^\ast$ is \emph{strictly locally testable} if
  there is $\ell \in \N_+$ and sets $\pi, \sigma \subseteq \Sigma^{\le \ell}$
  and $\mu \subseteq \Sigma^\ell$ such that, for every $w \in \Sigma^\ast$, $w
  \in L$ if and only if $p_{\ell-1}(w) \in \pi$, $I_\ell(w) \subseteq \mu$,  and
  $s_{\ell-1}(w) \in \sigma$.
  The class of all such languages is denoted by $\SLT$.

  A language $L \subseteq \Sigma^\ast$ is \emph{locally testable} if there is
  $\ell \in \N_+$ such that, for every $w_1, w_2 \in \Sigma^\ast$ with
  $p_{\ell-1}(w_1) = p_{\ell-1}(w_2)$, $I_\ell(w_1) = I_\ell(w_2)$, and
  $s_{\ell-1}(w_1) = s_{\ell-1}(w_2)$, we have that $w_1 \in L$ if and only if
  $w_2 \in L$.
  The class of locally testable languages is denoted by $\LT$.

  A language $L \subseteq \Sigma^\ast$ is \emph{locally threshold testable} if
  there are $\theta,\ell \in \N_+$ such that, for any two words $w_1, w_2 \in
  \Sigma^\ast$ for which the following conditions hold, $w_1 \in L$ if and only
  if $w_2 \in L$:
  \begin{enumerate}
    \item $p_{\ell-1}(w_1) = p_{\ell-1}(w_2)$ and $s_{\ell-1}(w_1) =
    s_{\ell-1}(w_2)$.
    \item For every $m \in \Sigma^\ell$, if $|w_i|_m < \theta$ for any $i \in
    \{1,2\}$, then $|w_1|_m = |w_2|_m$.
  \end{enumerate}
  The class of locally threshold testable languages is denoted by $\LTT$.
\end{definition}

The class $\LT$ equals the closure of $\SLT$ under Boolean operations (i.e.,
union, intersection, and complement) and the inclusion $\SLT \subsetneq \LT$ is
proper.
As for $\LTT$, it is well-known that it contains every language
\[
  \Th(m, \theta) = \{ w \in \Sigma^\ast \mid \abs{w}_m \le \theta \}
\]
where $m \in \Sigma^\ast$ and $\theta \in \N_0$.
Also, we have that $\LT \subsetneq \LTT$ and that $\LTT$ is closed under Boolean
operations.
We write $\SLT_\cup$ for the closure of $\SLT$.

\begin{definition}[LLT]%
  \label{def_llt}%
  For $\ell \in \N_0$, $\theta \in \R_0^+$, $\pi, \sigma \subseteq \Sigma^{\le
  \ell-1}$, and $\alpha\colon \Sigma^\ell \to \R_0^+$, 
  $\LLin_\ell(\pi,\sigma,\alpha,\theta)$ denotes the language of all $w \in
  \Sigma^+$ that satisfy $p_{\ell-1}(w) \in \pi$, $s_{\ell-1}(w) \in \sigma$,
  and
  \[
    \sum_{m \in \Sigma^\ell} \alpha(m) \cdot \abs{w}_m \le \theta.
  \]
  A language $L \subseteq \Sigma^+$ is said to be \emph{locally linearly
  testable} if there are $\ell$, $\pi$, $\sigma$, $\alpha$, and $\theta$ as
  above such that $L = \LLin_\ell(\pi,\sigma,\alpha,\theta)$.
  We denote the class of all such languages by $\LLT$.
\end{definition}

We write $\LLT_\cup$ for the closure of $\LLT$ under union and $\LLT_{\cup\cap}$
for its closure under both union and intersection.

\begin{proposition}[restate=restatepropLLTStructure,name=]%
  \label{prop_llt_structure}%
  The following hold (see \cref{fig_diag_llt}):
  \begin{enumerate}
    \item $\SLT \subsetneq \LLT \subsetneq \LLT_\cup \subseteq \LLT_{\cup\cap}
    \subsetneq \LTT$ and $\SLT_\cup \subsetneq \LLT_\cup$.
    \item $\LLT$ and $\LT$ as well as $\LLT_\cup$ and $\LT$ are incomparable.
    \item $\LTT$ equals the Boolean closure of $\LLT$.
  \end{enumerate}
\end{proposition}

\begin{figure}
  \centering
    \begin{tikzpicture}[every node/.style={minimum width=1.15cm}]
  \node (SLTc) {$\SLT_\cup$};
  \node (SLT) [left =of SLTc] {$\SLT$};
  \node (LT) [right =of SLTc] {$\LT$};
  \node (LTT) [right =of LT] {$\LTT$};
  \node (REG) [right =of LTT] {$\REG$};
  \node (LLT) [above =of SLTc] {$\LLT$};
  \node (LLTc) [above =of LT] {$\LLT_\cup$};
  \node (LLTcc) [above =of LTT] {$\LLT_{\cup\cap}$};

  \draw[->,>=latex'] (SLT) -- (SLTc);
  \draw[->,>=latex'] (SLTc) -- (LT);
  \draw[->,>=latex'] (LT) -- (LTT);
  \draw[->,>=latex'] (SLT) -- (LLT);
  \draw[->,>=latex'] (SLTc) -- (LLTc);
  \draw[->,>=latex'] (LLT) -- (LLTc);
  \draw[->,>=latex'] (LLTc) -- (LLTcc) node [midway,above] {?};
  \draw[->,>=latex'] (LLTcc) -- (LTT);
  \draw[->,>=latex'] (LTT) -- (REG);

  \draw[dashed] (LT) -- (LLT);
  \draw[dashed] (LT) -- (LLTc);
\end{tikzpicture}
  \caption{Placement of $\LLT$ in the subregular hierarchy.
    Arrows indicate inclusion relations; all are known to be strict except for
    the marked one (i.e., $\LLT_\cup \subseteq \LLT_{\cup\cap}$).
    Dashed lines between two classes denote they are incomparable.
  }
  \label{fig_diag_llt}
\end{figure}

Some relations between the classes are still open (see \cref{fig_diag_llt}) and
are left as a topic for future work.

\begin{proof}
  For the first item, note the inclusions $\LLT \subseteq \LLT_\cup$ and
  $\LLT_\cup \subseteq \LLT_{\cup\cap}$ are trivial, so we need only prove $\SLT
  \subseteq \LLT \subseteq \LTT$.
  The inclusion $\LLT_{\cup\cap} \subseteq \LTT$ then follows from $\LTT$ being
  closed under union and intersection, while $\SLT_\cup \subseteq \LLT_\cup$
  follows directly from $\SLT \subseteq \LLT$.
  The strictness of the inclusions follow all from the second item in the claim
  (since $\LLT$ then contains some language $L \notin \LT$, which is certainly
  not in $\SLT_\cup \subseteq \LT$); the only exceptions are $\LLT \subsetneq
  \LLT_\cup$, which we address further below, and $\LLT_{\cup\cap} \subsetneq
  \LTT$, which is proved in \cref{thm_char_twosided_PACA} (and does not depend
  on the results here).

  The first inclusion $\SLT \subseteq \LLT$ is easiest.
  Given $L \in \SLT$, we know there is $\ell \in \N_+$ so that $L$ is defined
  based on sets of allowed prefixes $\pi \subseteq \Sigma^{\le \ell}$, infixes
  $\mu \subseteq \Sigma^\ell$, and suffixes $\sigma \subseteq \Sigma^{\le
  \ell}$.
  Clearly we have then $L = \LLin_\ell(\pi,\sigma,\alpha,\theta)$ for $\theta =
  1/2$ and
  \[
    \alpha(m) = \begin{cases}
      0, & m \in \mu \\
      1, & m \notin \mu.
    \end{cases}
  \]

  For the inclusion $\LLT \subseteq \LTT$, let $L =
  \LLin_\ell(\pi,\sigma,\alpha,\theta)$ be given.
  The proof is by induction on the number $k$ of words $m \in \Sigma^\ell$ such
  that $\alpha(m) \neq 0$.
  If $k = 0$, then the linear condition of $L$ is always satisfied, directly
  implying $L \in \LTT$ (or, better yet, $L \in \SLT$).
  For the induction step, suppose the claim has been proven up to a number $k$,
  and let there be $k+1$ words $m \in \Sigma^\ell$ with $\alpha(m) \neq 0$.
  In addition, let $\mu \in \Sigma^\ell$ with $\alpha(\mu) \neq 0$ be arbitrary,
  and let $r \in \N_0$ be maximal with $r \alpha(\mu) \le \theta$.
  Consider the languages $L_i = \LTT_\ell(\pi,\sigma,\alpha_i,\theta_i)$ for $i
  \in [0,r]$, $\theta_i = \theta - i \alpha(\mu)$, and $\alpha_i$ such that, for
  every $m \in \Sigma^\ell$,
  \[
    \alpha_i(m) = \begin{cases}
      \alpha(m), & m \neq \mu; \\
      0, & m = \mu.
    \end{cases}
  \]
  By the induction hypothesis, the $L_i$ are all in $\LTT$.
  On the other hand, we have
  \[
    L = \bigcup_{i=0}^r {(\Th(\mu,i) \cap L_i) \setminus \Th(\mu,i-1)},
  \]
  which implies $L \in \LTT$ since $\LTT$ is closed under Boolean operations.

  For the strictness of the inclusion $\LLT \subsetneq \LLT_\cup$, consider the
  language $L = \Th(1,1) \cup \Th(2,1)$ over the ternary alphabet $\Sigma = \{
  0,1,2 \}$.
  Clearly we have $L \in \LLT_\cup$, so assume that $L =
  \LLin_\ell(\pi,\sigma,\alpha,\theta)$ holds for some choice of $\ell$, $\pi$,
  $\sigma$, $\alpha$, and $\theta$.
  Then necessarily $\alpha(w) > 0$ for some $w \in \binalph^\ell$ with
  $\abs{w}_1 = 1$ since otherwise we would have a contradiction to $0^\ell 1
  0^\ell (2 0^\ell)^2 \in L$ and $0^\ell (1 0^\ell)^2 (2 0^\ell)^2 \notin L$.
  However, we then get that $0^\ell (1 0^\ell)^{\ceil{1/\alpha(w)}} 2 0^\ell
  \notin L$, which is a contradiction.

  For the second item in the claim, we show there are $L_1$ and $L_2$ such that
  $L_1 \in \LLT \setminus \LT$ and $L_2 \in \LT \setminus \LLT_\cup$.
  For the first one, probably the simplest is to set $L_1 = \Th(1,1)$ (where the
  underlying alphabet is $\Sigma = \binalph$).
  The language is clearly not in $\LT$ (because otherwise either \emph{both}
  $0^n10^n$ and $0^n10^n10^n$ would be in $L_1$ or not (for sufficiently large
  $n$)); meanwhile, we have $L_1 \in \LLT$ since $\ell = 1$, $\pi = \sigma = \{
  \varepsilon \}$, $\alpha(0) = 0$, $\alpha(1) = 1$, and $\theta = 1$ satisfy
  $L_1 = \LLin_\ell(\pi,\sigma,\alpha,\theta)$.
  As for $L_2$, consider
  \[
    L_2 = \{ w \in \binalph^\ast \mid w \notin \{0\}^\ast \},
  \]
  which is in $\LT$ because $L_2 = \binalph^\ast \setminus \{0\}^\ast$ (and
  $\LT$ is closed under Boolean operations).
  To argue that $L_2 \notin \LLT_\cup$, suppose towards a contradiction that
  $L_2 = \bigcup_{i=1}^k L_i'$ for $L_i' =
  \LLin_{\ell_i}(\pi_i,\sigma_i,\alpha_i,\theta_i)$.
  Since $k$ is finite, there must be at least one $L_i'$ that contains
  infinitely many words of the form $0^j10^j$ for $j \in \N_0$.
  This implies $0^{\ell_i-1} \in \pi_j,\sigma_j$ as well as $\alpha(0^{\ell_i})
  = 0$, from which it follows that
  \[
    \sum_{m \in \binalph^{\ell_i}} \alpha_i(m) \cdot \abs{0^k}_m
    = \alpha_i(0^{\ell_i}) \cdot \abs{0^k}_{0^{\ell_i}}
    = 0
    \le \theta_i
  \]
  for any $k \in \N_0$, thus contradicting $0^k \notin L_i'$.

  Finally, the third item directly follows from the following well-known
  characterization of $\LTT$:
  A language $L$ is in $\LTT$ if and only if it can be expressed as the Boolean
  combination of languages of the form $\Th(m,\theta)$, $\pi \Sigma^\ast$, and
  $\Sigma^\ast \sigma$ where $m, \pi, \sigma \in \Sigma^\ast$ and $\theta \in
  \N_0$.
  Obviously these three types of languages are all contained in $\LLT$; since
  $\LLT \subseteq \LTT$, it follows that $\LTT$ equals the Boolean closure of
  $\LLT$.
\end{proof}

\subsubsection{The Proof}
\label{sec_proof_thm_char_twosided}

With the terminology of \cref{sec_local_languages} in place, we now turn to:

\restatethmCharTwosided*

The theorem is proven by showing the two inclusions.
We first give a brief overview of the ideas involved.

\subparagraph{First inclusion.}
The first step is showing 
that we can \enquote{tweak} the components of the $\LLT$ condition so that it is
more amenable to being tested by a PACA (\cref{lem_simpl_LLin}).
In particular, we prove we can assume the $\alpha(m)$ weights are such that
$2^{-\alpha(m)}$ can be represented using a constant number of bits.
Having done so, the construction is more or less straightforward:
We collect the subwords of length $\ell$ in every cell and then accept with the
\enquote{correct} probabilities; that is, if a cell sees a subword $m$, then it
accepts with probability $2^{-\alpha(m)}$.
To lift the inclusion from $\LLT$ to $\LLT_{\cup\cap}$, we show that the class
of languages that are recognizable by constant-time two-sided error PACAs is
closed under union and intersection (\cref{prop_closure_union}).

\subparagraph{Second inclusion.}
The second inclusion
(i.e., showing that $L(C) \in \LTT$ for every constant-time two-sided error
PACA $C$)
is considerably more complex.
Let $C$ be a PACA with time complexity at most $T$.
The proof again bases on the class $\LLT$ and uses the fact that $\LTT$ equals
the Boolean closure over $\LLT$ (which we show as a separate result):
\begin{enumerate}
  \item As a warm-up, we consider the case where the cells of $C$ accept all
  independently from one another.
  In addition, we assume that, if $C$ accepts, then it does so in a fixed time
  step $t < T$.
  The argument is then relatively straightforward since we need only consider
  subwords of length $\ell = 2T+1$ and set their $\LLT$ weight according to the
  acceptance probability that the cell in the middle of the subword would have
  in $C$.
  \item Next we relax the requirement on independence between the cells (but
  still assume a fixed time step for acceptance).
  The situation then requires quite a bit of care since the $\LLT$ condition
  does not account for subword overlaps at all.
  For instance, there may be cells $c_1$ and $c_2$ that are further than $T$
  cells apart and that both accept with non-zero probability but where $c_i$
  accepts if and only if $c_{3-i}$ does not (see
  \cref{appx_example_daca_vs_paca}
  for a concrete example).
  We solve this issue by blowing up $\ell$ so that a subword covers not only a
  single cell's neighborhood but that of an \emph{entire group} of cells whose behavior may be correlated with one other.
  Here we once more resort to \cref{lem_critical_cells,lem_independence} to
  upper-bound the size of this neighborhood by a constant.
  \item Finally, we generalize what we have shown so that it also holds in the
  case where $C$ may accept in any step $t < T$.
  This is the only part in the proof where closure under complement is required.
  The argument bases on generalizing the ideas of the previous step to the case
  where the automaton may accept in multiple time steps and then applying the
  inclusion-exclusion principle.
\end{enumerate}

We now elaborate on these ideas.
As discussed above, we first need to make the linear condition of $\LLT$ a bit
more manageable:

\begin{lemma}[restate=restatelemSimplLLin,name=]%
  \label{lem_simpl_LLin}%
  For any $L = \LLin_\ell(\pi,\sigma,\alpha,\theta)$,
  there is $\alpha'$ such that $L = \LLin_\ell(\pi,\sigma,\alpha',\theta')$ and:
  \begin{enumerate}
    \item The threshold $\theta'$ is equal to $1$.
    \item There is a constant $\eps > 0$ such that, for every $w \in
    \Sigma^\ast$, we have either $f'(w) < 1-\eps$ or $f'(w) > 1+\eps$, where
    $f'(w) = \sum_{m \in \Sigma^\ell} \alpha'(m) \cdot \abs{w}_m$.
    \item There is $k \in \N_0$ such that, for every $m$, there is $n \in [2^k]$
    such that $\alpha'(m) = k - \log(n+1)$.
  \end{enumerate}
\end{lemma}

The first two items ensure the sum $f'(w)$ is always \enquote{far away} from
$1$.
In turn, the third item enables us to represent $2^{-\alpha'(m)}$ using at most
$k$ (and, in particular, $O(1)$ many) bits.

\begin{proof}
  The case where $\alpha(m) = 0$ for every $m$ is trivial, so suppose there is
  some $m$ for which $\alpha(m) > 0$.
  Because $\abs{w}_m \in \N_0$ for every $w$ and every $m$, $f(w) = \sum_{m \in
    \Sigma^\ell} \alpha(m) \cdot \abs{w}_m$ is such that, given any $b \in
  \R_0^+$, there are only finitely many values of $f(w) \le b$ (i.e., the set
  $\{ f(w) \mid f(w) \le b \}$ is finite).
  Hence, by setting
  \[
    r = \frac{1}{2} \left( \max_{f(w) \le \theta}{f(w)}
      + \min_{f(w) > \theta}{f(w)}
    \right)
  \]
  and $\alpha''(m) = \alpha(m) / r$, we have that $f''(w) = \sum_{m \in
    \Sigma^\ell} \alpha''(m) \cdot \abs{w}_m$ satisfies the second condition
  from the claim (i.e., for every $w \in \Sigma^\ast$, either $f''(w) < 1-\eps$
  or $f''(w) > 1+\eps$) for some adequate choice of $\eps > 0$.

  Now we show how to satisfy the last condition without violating the first two.
  Essentially, we will choose $k$ to be sufficiently large and then
  \enquote{round up} each $\alpha''(m)$ to the nearest value of the form $k -
  \log(n+1)$.
  To that end, let $a$ be the minimal $\alpha''(m)$ for which $\alpha''(m) > 0$.
  In addition, let $k \in \N_0$ be such that $\alpha''(m) < k/2$ for every $m$
  and that
  \[
    \log(2^{k/2} + 1) - \log(2^{k/2})
    = \log(2^{k/2} + 1) - \frac{k}{2}
    \le \frac{a\eps}{2\abs{\Sigma}^\ell}.
  \]
  (This is possible because $\log(n+1) - \log n$ tends to zero as $n \to
  \infty$ and the right-hand side is constant.)
  Then, for every $m$, we set $\alpha'(m) = k - \log(n+1)$ where $n \in [2^k]$
  is maximal such that $\alpha'(m) \ge \alpha''(m)$.
  By construction, $f'(w) \ge f''(w)$, so we need only argue that there is
  $\eps'>0$ such that, for every $w$ for which $f''(w) < 1-\eps$, we also have
  $f'(w) <  1-\eps'$.
  In particular every said $w$ must be such that, for every $m$, $\abs{w}_m
  \le 1/a$ (otherwise we would have $f''(w) > 1$).
  Noting that $\abs{\alpha'(m) - \alpha''(m)}$ is maximal when $\alpha'(m) =
  k/2$ and $\alpha''(m) = k - \log(2^{k/2}+1) + \delta$ for very small $\delta >
  0$, we observe that
  \[
    f'(w) = \sum_{m \in \Sigma^\ell} \alpha'(m) \cdot \abs{w}_m
    < f''(w) + \frac{\abs{\Sigma}^\ell}{a}
      \left( \log(2^{k/2} + 1) - \frac{k}{2} \right)
    < 1-\eps + \frac{\eps}{2}
    = 1 - \frac{\eps}{2},
  \]
  that is, $f'(w) < 1 - \eps'$ for $\eps' = \eps/2$, as desired.
\end{proof}

We also need the following result:

\begin{proposition}[restate=restatePropClosureUnion,name=]%
  \label{prop_closure_union}%
  Let $C_1$ and $C_2$ be constant-time two-sided error PACA.
  Then there are constant-time two-sided error PACA $C_\cup$ and $C_\cap$ such
  that $L(C_\cup) = L(C_1) \cup L(C_2)$ and $C_\cap = L(C_1) \cap L(C_2)$.
\end{proposition}

\begin{proof}
  We first address the closure under union.
  Using \cref{prop_twosided_error_reduction}, we may assume that $C_1$ and $C_2$
  are two-sided $\eps$-error PACA for some $\eps < 1/6$.
  As we show further below, it suffices to have $C_\cup$ simulate both $C_1$ and
  $C_2$ simultaneously (using independent random bits) and accept if either of
  them does.
  To realize this, we can simply adapt the construction from
  \cref{prop_onesided_error_reduction} using $m = 2$ to simulate one copy of
  $C_1$ and $C_2$ each (instead of two independent copies of the same PACA).

  The probability that $C_\cup$ accepts an $x \in L_i \setminus L_j$ for $i,j
  \in \{1,2\}$ and $i \neq j$ is at least $1-\eps$, and the probability that
  $C_\cup$ accepts $x \in L_1 \cap L_2$ is even larger (i.e., at least $1-\eps^2
  \ge 1-\eps$).
  Conversely, the probability that $C_\cup$ accepts an $x \notin L_1 \cup L_2$
  is
  \[
    \Pr[\text{$C_\cup$ accepts $x$}]
    \le \Pr[\text{$C_1$ accepts $x$}] + \Pr[\text{$C_2$ accepts $x$}]
    = 2\eps < \frac{1}{3}.
  \]
  Since $C_\cup$ has time complexity $2T = O(T)$, the claim follows.

  For the closure under intersection, we will use a similar strategy.
  This time suppose that the error probability $\eps$ of $C_1$ and $C_2$ is so
  that $(1-\eps)^2 \ge 2/3$ (e.g., $\eps = 1/10$ suffices).
  Again we let $C_\cap$ simulate two copies of $C_1$ of $C_2$ simultaneously;
  however, this time we use a simplified version of the construction from the
  proof of \cref{prop_twosided_error_reduction}.
  More specifically, we set $m = 2$ and leave out the $[M]$ component from the
  construction.
  That is, $C_\cap$ randomly (and independently) picks random inputs $r_1$ and
  $r_2$ for $C_1$ and $C_2$, respectively, and accepts if and only if $C_1$
  accepts with coin tosses from $r_1$ and also $C_2$ accepts with coin tosses
  from $r_2$.
  (One must also be careful since $C_1$ and $C_2$ may accept in different time
  steps, but this is already accounted for in the construction from
  \cref{prop_twosided_error_reduction}.)

  Since the copies of $C_1$ and $C_2$ are simulated independently from one
  another and $C_\cap$ accepts if and only both do (in the simulation), we have
  \[
    \Pr[\text{$C_\cap$ accepts $x$}]
    = \Pr[\text{$C_1$ accepts $x$}] \Pr[\text{$C_2$ accepts $x$}].
  \]
  In case $x \in L(C_1) \cap L(C_2)$, this probability is at least $(1-\eps)^2
  \ge 2/3$; otherwise it is upper-bounded by $\eps < 1/3$.
  Since $C$ has time complexity $O(T^2)$, which is constant, the claim follows.
\end{proof}

We are now in position to prove \cref{thm_char_twosided_PACA}.

\begin{proof}[Proof of \cref{thm_char_twosided_PACA}]
  We prove the two inclusions from the theorem's statement.
  The first one we address is that of $\LLT_{\cup\cap}$ in the class of
  constant-time two-sided error PACA.

  \proofsubparagraph{First inclusion.}
  Given $L = \LLin_\ell(\pi,\sigma,\alpha,\theta)$, we construct a constant-time
  two-sided error PACA $C$ with $L(C) = L$.
  This suffices since by the closure properties shown in
  \cref{prop_closure_union}.
  We apply \cref{lem_simpl_LLin} and assume $\theta=1$ and that there are $k$
  and $\eps$ as in the statement of \cref{lem_simpl_LLin}.
  For simplicity, we also assume $\ell = 2k+1$.

  The automaton $C$ operates in $k$ steps as follows:
  Every cell sends its input symbol in both directions as a signal and, at the
  same time, aggregates the symbols it sees, thus allowing it to determine the
  initial configuration $m \in \Sigma^\ell$ of its $k$-neighborhood.
  Meanwhile, every cell also collects $k$ random bits $r \in \binalph^k$.
  The decision to accept is then simultaneously made in the $k$-th step, where
  a cell with $k$-neighborhood $m$ accepts with probability $2^{-\alpha(m)}$
  (independently of other cells).
  (This can be realized, for instance, by seeing $r$ as the representation of a
  $k$-bit integer in $[2^k]$ and accepting if and only if $r \le n$, where $n$
  is such that $2^{-\alpha(m)} = (n+1)/2^k$.)
  In the case of the first (resp., last) cell of $C$, it also checks that the
  prefix (resp., suffix) of the input is in $\pi$ (resp., $\sigma$), rejecting
  unconditionally if this is not the case.

  Hence, for an input word $w \in \Sigma^+$, the probability that $C$ accepts is
  \[
    \prod_{m \in \Sigma^\ell} \left(\frac{1}{2^{\alpha(m)}}\right)^{\abs{w}_m}
    = 2^{-f(w)},
  \]
  where $f(w) = \sum_{m \in \Sigma^\ell} \alpha(m) \cdot \abs{w}_m$.
  Thus, if $w \in L$, then $C$ accepts with probability $2^{-f(w)} >
  (1/2)^{1-\eps}$; conversely, if $w \notin L$, the probability that $C$ accepts
  is $2^{-f(w)} < (1/2)^{1+\eps}$.
  Since $\eps$ is constant, we may apply \cref{prop_twosided_error_reduction}
  and reduce the error to $1/3$.

  \proofsubparagraph{Second inclusion.}
  The proof of the second inclusion is more involved.
  Let $C$ be a $T$-time two-sided error PACA for some $T \in \N_+$.
  We shall obtain $L(C) \in \LTT$ in three steps:
  \begin{enumerate}
    \item The first step is a warm-up where the cells of $C$ accept all
    independently from one another and that, if $C$ accepts, then it does so in
    a fixed time step $t < T$.
    \item Next we relax the requirement on independence between the cells by
    considering groups of cells (of maximal size) that may be correlated with
    one another regarding their acceptance.
    \item Finally, we generalize what we have shown so that it also holds in the
    case where $C$ may accept in any step $t < T$.
    This is the only part in the proof where closure under complement is
    required.
    (Here we use item $3$ of \cref{prop_llt_structure}.)
  \end{enumerate}

  \begin{description}
    \item[Step 1.] Suppose that $C$ only accepts in a fixed time step $t < T$
    and that the events of any two cells accepting are independent from one
    another.
    We show that $L(C) = \LLin(\pi,\sigma,\alpha,\theta)$ for an adequate choice
    of parameters.
     
    Set $\ell = 2t+1$ and let $p_m$ be the probability that a
    cell with $t$-neighborhood $m \in \Sigma^\ell$ accepts in step $t$. 
    In addition, let $\pi = \{ p_{\ell - 1}(w) \mid w \in L(C) \}$ and $\sigma =
    \{ s_{\ell - 1}(w) \mid w \in L(C) \}$ as well as $\theta = \log(3/2)$ and
    $\alpha(m) = \log(1/p_m)$ for $m \in \Sigma^\ell$.
    The probability that $C$ accepts a word $w \in \Sigma^+$ is
    \[
      \prod_{m \in \Sigma^\ell} p_m^{\abs{w}_m}
      = \prod_{m \in \Sigma^\ell}
        \left(\frac{1}{2^{\alpha(m)}}\right)^{\abs{w}_m}
      = 2^{-f(w)},
    \]
    which is at least $2/3$ if and only if $f(w) \le \theta$.
    It follows that $L(C) = \LLin_\ell(\pi,\sigma,\alpha,\theta)$.

    \item[Step 2.] We now relax the requirements from the previous step so that
    the events of any two cells accepting need no longer be independent from one
    another.
    ($C$ still only accepts in the fixed time step $t$.)
     
    Let $K = 2^{O(T^2)}$ be the upper bound from \cref{lem_critical_cells} and
    $\ell = 2(K+2)T$.
    Again, we set $\pi = \{ p_{\ell - 1}(w) \mid w \in L(C) \}$, $\sigma = \{
    s_{\ell - 1}(w) \mid w \in L(C) \}$, and $\theta = \log(3/2)$.
    As for $\alpha(m)$, we set $\alpha(m) = 0$ unless $m$ is such that there is
    $d \le K$ with
    \[
      m = abr_1s_1r_2s_2 \cdots r_ds_dc
    \]
    where $a \in \Sigma^T$ is arbitrary, $b \in \Sigma^{2T}$ contains \emph{no
    critical cells}, each $r_j \in \Sigma$ is a critical cell (for step $t$),
    the $s_j \in \Sigma^{\le 2T-1}$ are arbitrary, and $c \in \Sigma^\ast$ has
    length $\abs{c} \ge T$ and, if $c$ contains any critical cell, then this
    cell accepts independently from $r_d$.
    In addition, we require $c$ to be of maximal length with this property.

    Note we need $a$ as context to ensure that $b$ indeed does not contain 
    critical cells (since determining this requires knowledge of the states in
    the $T$-neighborhood of the respective cell); the same holds for $r_d$ and
    $c$.
    By construction and \cref{lem_independence}, the group of cells
    $r_1,\dots,r_d$ is such that (although its cells are not necessarily 
    independent from one another) its cells accepts independently from any other
    critical cell in $C$.
    Furthermore, $b$ ensures $m$ aligns properly with the group and that the
    group does not appear in any other infix.

    Letting $p_m$ be the probability that every one of the $r_j$ accept, for $m$
    as above we set $\alpha(m) = \log(1/p_m)$.
    Then, as before, the probability that $C$ accepts a word $w \in \Sigma^+$ is
    $2^{-f(w)}$, which is at least $2/3$ if and only if $f(w) \le \theta$, thus
    implying $L(C) = \LLin_\ell(\pi,\sigma,\alpha,\theta)$.

    \item[Step 3.] In this final step we generalize the argument so it also
    applies to the case where $C$ may accept in any time step $t < T$.
    First note that, given any $p>0$, if we set $\theta = \log(1/p)$ in the
    second step above (instead of $\log(3/2)$), then we have actually shown that
    \[
      \{ w \in \Sigma^+ \mid \Pr[\text{$C$ accepts $w$ in step $t$}] \ge p \}
      = \LLin_\ell(\pi,\sigma,\alpha,\theta).
    \]
    In fact, we can generalize this even further:
    Given any $\varnothing \neq \tau \subseteq [T]$, by setting $\alpha$
    adequately we can consider the acceptance probability for the steps in
    $\tau$ altogether:\footnote{%
    Of course we are being a bit sloppy here since \cref{def_PACA} demands that
    a PACA should halt whenever it accepts.
    What is actually meant is that, having fixed some random input, if we extend
    the space-time diagram of $C$ on input $w$ so that it spans all of its first
    $T$ steps (simply by applying the transition function of $C$), then, for
    every $t \in \tau$, the $t$-th line in the diagram contains only accepting
    cells.}
    \[
      L(\tau,p)
      = \{ w \in \Sigma^+ \mid
        \Pr[\text{$C$ accepts $w$ in every step $t \in \tau$}] \ge p \}
      = \LLin_\ell(\pi,\sigma,\alpha,\theta).
    \]
    This is because the bound on critical cells of \cref{lem_critical_cells}
    holds for all steps where $C$ accepts with non-zero probability and, in
    addition, as defined above $m$ already gives enough context to check if the
    respective critical cells also accept in any previous step.
    (That is, we construct $m$ as above by using $t = \max \tau$; however, since
    the sets of critical cells for different time steps may not be identical, we
    must also relax the condition for the $r_i$ so that $r_i$ need only be a
    critical cell in at least one of the time steps of $\tau$.)
    Since $\LTT$ is closed under complement, we then also have
    \[
      \overline{L(\tau,p)}
      = \{ w \in \Sigma^+ \mid
        \Pr[\text{$C$ accepts $w$ in every step $t \in \tau$}] < p \}
      \in \LTT.
    \]

    Fix some input word $w \in \Sigma^+$ to $C$.
    For $\varnothing \neq \tau \subseteq [T]$, let $Z_\tau$ denote the event
    where $C$ accepts $w$ in every step $t \in \tau$.
    By the inclusion-exclusion principle, we have
    \[
      \Pr[\text{$C$ accepts $w$}]
      = \Pr\left[ \exists t \in [T]: Z_{\{t\}} \right]
      = \sum_{\stackrel{\tau \subseteq [T]}{\abs{\tau} = 1}} \Pr[Z_\tau]
        - \sum_{\stackrel{\tau \subseteq [T]}{\abs{\tau} = 2}} \Pr[Z_\tau]
        + \cdots
        + (-1)^{T+1} \Pr[Z_{[T]}]
    \]
    (where the probabilities are taken over the coin tosses of $C$).
    This means that, if we are somehow given values for $p(\tau) = \Pr[Z_\tau]$
    so that the sum above is at least $2/3$, then we can intersect a finite
    number of $L(\tau,p(\tau))$ languages and their complements and obtain some
    language that is guaranteed to contain only words in $L(C)$.
    Concretely, let $p(\tau) \ge 0$ for every $\varnothing \neq \tau \subseteq
    [T]$ be given so that
    \[
      \sum_{\stackrel{\varnothing \neq \tau \subseteq [T]}
          {\text{$\abs{\tau}$ odd}}}
        p(\tau)
      - \sum_{\stackrel{\varnothing \neq \tau \subseteq [T]}
          {\text{$\abs{\tau}$ even}}}
        p(\tau)
      \ge \frac{2}{3}.
    \]
    Let
    \[
      \mathcal{T}_{\textrm{odd}} = \{ \tau \subseteq [T]
        \mid \tau \neq \varnothing, \, \text{$\abs{\tau}$ odd}, \,
        p(\tau) > 0 \}
    \]
    and similarly
    \[
      \mathcal{T}_{\textrm{even}} = \{ \tau \subseteq [T]
        \mid \tau \neq \varnothing, \, \text{$\abs{\tau}$ even}, \,
        p(\tau) > 0 \}.
    \]
    Then necessarily
    \[
      L(p) = \left(
        \bigcap_{\tau \in \mathcal{T}_\textrm{odd}} L(\tau,p(\tau))
      \right)
      \cap \left(
        \bigcap_{\tau \in \mathcal{T}_\textrm{even}} \overline{L(\tau,p(\tau))}
      \right)
      \subseteq L(C)
    \]
    contains every $w \in L(C)$ for which $\Pr[Z_\tau] \ge p(\tau)$ for $\tau
    \in \mathcal{T}_{\mathrm{odd}}$ and $\Pr[Z_\tau] \le p(\tau)$ for $\tau \in
    \mathcal{T}_{\mathrm{even}}$.
 
    The key observation is that \emph{there are only finitely many values the
    $\Pr[Z_\tau]$ may assume}.
    This is because $Z_\tau$ only depends on a finite number of coin tosses,
    namely the ones in the lightcones of the cells that are critical in at least
    one of the steps in $\tau$ (which, again, is finite due to
    \cref{lem_critical_cells}).
    Hence, letting $P$ denote the set of all possible mappings of the $\tau$
    subsets to these values, we may write
    \[
      L(C) = \bigcup_{p \in P} L(p) \in \LTT.
    \]

    \proofsubparagraph{Strictness of inclusion.}
    The final statement left to prove is that the inclusion just proven is
    proper.
    This is comparatively much simpler to prove.
    We show that the language
    \[
      L = \{ w \in \binalph^+ \mid \abs{w}_1 \ge 2 \} \in \LTT
    \]
    cannot be accepted by two-sided error PACA in constant time.
    
    For the sake of argument, assume there is such a PACA $C$ with time
    complexity $T \in \N_+$.
    Consider which cells in $C$ are critical based on their initial local
    configuration.
    Certainly a cell with an all-zeroes configuration $0^{2T-1}$ cannot be
    critical.
    Since $0^n10^n \notin L(C)$ for any $n$ (but $0^n10^n1 \in L(C)$), there
    must be $m_1,m_2$ so that $m_1 + m_2 = 2T-2$ and $c = 0^{m_1}10^{m_2}$ is
    the initial local configuration of a critical cell.
    This means that in
    \[
      x = 0^{2T} (c 0^{2T})^{T2^T} \in L
    \]
    we have at least $2^T$ cells in $x$ that are critical for the same time step
    $t \in [T]$ (by an averaging argument) and that are also all independent
    from one another (by \cref{lem_independence}).
    In turn, this implies the following, which contradicts $x \in L(C)$:
    \[
      \Pr[\text{$C$ accepts $x$}] \le \left( 1 - 2^{-T} \right)^{2^T}
      < \frac{1}{e} < \frac{2}{3}.
      \qedhere
    \]
  \end{description}
\end{proof}


\section{The General Sublinear-Time Case}
\label{sec_sublinear_time}

Recall we say a DACA $C$ is \emph{equivalent} to a PACA $C'$ if $L(C) = L(C')$.
In this section, we recall and briefly discuss:

\restatethmPACAEqACA*
\vspace{\topsep} 

To obtain \cref{thm_PACA_eq_ACA}, we prove the following more general result:

\begin{proposition}[restate=restatePropPACAInclusionACA,label=]%
  \label{prop_PACA_inclusion_ACA}%
  There is a constant $c>0$ such that the following holds:
  Let monotone functions $T,T',h,p\colon \N_+ \to \N_+$ be given with $h(n) \le
  2^n$, $p(n) = \poly(n)$, and $p(n) \ge n$ and such that, for every $n$ and
  $p'(n) = \Theta(p(n) \log h(n))$,
  \[
    T(6h(n)p(n)) \ge c p'(n).
  \]
  In addition, for $n$ given in unary, let the binary representation of $h(n)$
  and $p'(n)$ be computable in $O(p'(n))$ time by a Turing machine and, for $N$
  given in unary, let $T'(N)$ be computable in $O(T'(N))$ time by a Turing
  machine.
  Furthermore, suppose that, for every $T$-time one-sided error PACA $C$, there
  is a $T'$-time DACA $C'$ such that $L(C) = L(C')$.
  Then
  \[
    \RTIME[p(n)] \subseteq
    \TIME[h(n) \, \cdot \, T'(h(n)\!\cdot\!\poly(n)) \, \cdot \, \poly(n)].
  \]
\end{proposition}

The first item of \cref{thm_PACA_eq_ACA} is obtained by setting (say) $T(n) =
n^\eps$, $T'(n) = n^d$, and $h(n) = p(n)^{2(1/\eps-1)}$ (assuming $\eps < 1$).
For the second one, letting $p(n) = n^a$ where $a > 0$ is arbitrary and (again)
$T'(n) = n^d$, set $T(n) = (\log n)^{2+a/\eps}$ and $h(n) = 2^{n^\eps/(d+1)}$.
Finally, for the last one, letting again $p(n) = n^a$ and $T'(n) = n^d$, set
$T(n) = (\log n)^b$ and $h(n) = 2^{n^c}$.

At the core of the proof of \cref{prop_PACA_inclusion_ACA} is a padding
argument.
Nevertheless, we cannot stress enough that the padding itself is highly
nontrivial.
In particular, it requires a clever implementation that ensures it can be
verified \emph{in parallel} and also \emph{without initial knowledge of the
input length}.
To see why this is so, observe that, if we simply use a \enquote{standard} form
of padding where we map $x \in \binalph^+$ to $x' = x 0^{p(\abs{x})}$ (where
$p\colon \N_+ \to \N_0$ gives the desired padding length), then it is impossible
for the automaton to distinguish between this and, say, $x'' = x
0^{p(\abs{x})/2}$ in $o(p)$ time (assuming, e.g., $p(\abs{x}) =
\Omega(\abs{x})$).
The reason for this is that, since cells are initially completely unaware of
their position in the input, the cells with an all-zeroes neighborhood must
behave exactly the same.
More specifically, we can use an argument as in the proof of
\cref{thm_one_vs_twosided} to show that the automaton must behave the same on
both $x'$ and $x''$ (in the sense that it accepts the one if and only if it
accepts the other) unless it \enquote{looks at the whole input} (i.e., unless it
has $\Omega(p(\abs{x})$ time complexity).

The padding technique we use can be traced back to
\cite{ibarra85_fast_tcs}.
In a nutshell, we split the input into blocks of the same size that redundantly
encode the input length in a locally verifiable way.
More importantly, the blocks are numbered from left to right in ascending order,
which also allows us to verify that we have the number of blocks that we need.
This is crucial in order to ensure the input is \enquote{long enough} and the
PACA achieves the time complexity that we desire (as a function of the input
length).

\begin{proof}
  Let $L \in \RP$ be decided by an $\RP$ machine $R$ whose running time is
  upper-bounded by $p$.
  Without restriction, we assume $p(n) \ge n$.
  Using standard error reduction in $\RP$, there is then an $\RP$ machine $R'$
  with running time $p'(n) = \Theta(p(n) \log h(n))$ (i.e., polynomial in $n$),
  space complexity at most $p(n)$, and which errs on $x \in L$ with probability
  strictly less than $1/2h(n)$.
  Based on $L$ we define the language
  \[
    L' = \{ \bin_n(0) \# x_0 \# 0^{p(n)}
      \% \cdots \%
      \bin_n(h(n)-1) \# x_{h(n)-1} \# 0^{p(n)}
      \mid n \in \N_+, x_i \in L \cap \Sigma^n \},
  \]
  where $\bin_n(i)$ denotes the $n$-bit representation of $i < 2^n$.
  Note the length of an instance of $L'$ is $N \le 6h(n)p(n) = O(h(n)
  \poly(n))$.

  We claim there is $c>0$ such that $L'$ can be accepted in at most $c p'(n) =
  O(p'(n))$ (and in particular less than $T(N)$) time by a one-sided error PACA
  $C$.
  The construction is relatively straightforward:
  We refer to each group of cells $\bin_n(i)\#x_i\#0^{p'(n)}$ separated by the
  $\%$ symbols as a \emph{block} and the three binary strings in each block
  (separated by the $\#$ symbols) as its \emph{components}.
  First each block $a_1\#a_2\#a_3$ checks that its components have correct
  sizes, that is, that $\abs{a_1} = \abs{a_2}$ and $\abs{a_3} = p(\abs{a_1})$.
  Then the block communicates with its right neighbor $b_1\#b_2\#b_3$ (if it
  exists) and checks that $\abs{a_i} = \abs{b_i}$ for every $i$ and that, if
  $a_1 = \bin_n(j)$, then $b_1 = \bin_n(j+1)$.
  In addition, the leftmost block checks that its first component is equal to
  $\bin_n(0)$; similarly, the rightmost block computes $h(n)$ (in $O(p'(n))$
  time) and checks that its first component is equal to $\bin_n(h(n)-1)$.
  Following these initial checks, each block then simulates $R'$ (using bits
  from its random input as needed) on the input given in its second component
  using its third component as the tape.
  If $R'$ accepts, then all cells in the respective block turn accepting.
  In addition, the delimiter $\%$ is always accepting unless it is a border
  cell.

  Clearly $C$ accepts if and only if its input is correctly formatted and $R'$
  accepts every one of the $x_i$ (conditioned on the coin tosses that are chosen
  for it by the respective cells of $C$).
  Using a union bound, the probability that $C$ errs on an input $x \in L'$ is
  \[
    \Pr[C(x,U_{T \times n}) = 0]
    \le \sum_{i=0}^{h(n)-1} \Pr[R(x_i) = 0]
    < \sum_{i=0}^{h(n)-1} \frac{1}{2h(n)}
    = \frac{1}{2}.
  \]
  In addition, the total running time of $C$ is the time needed for the
  syntactic checks (requiring $O(p'(n))$ time), plus the time spent simulating
  $R'$ (again, $O(p'(n))$ time using standard simulation techniques).
  Hence, we can implement $C$ so that it runs in at most $c p'(n)$ time
  for some constant $c>0$, as desired.

  Now suppose there is a DACA $C'$ equivalent to $C$ as in the statement of the
  theorem.
  We shall show there is a deterministic (single-tape) Turing machine that
  decides $L$ with the purported time complexity.
  Consider namely the machine $S$ which, on an input $x \in \binalph^n$ of $L$,
  produces the input
  \[
    x' = \bin_n(0) \# x \# 0^{p(n)} \% \cdots \% \bin_n(h(n)-1) \# x \# 0^{p(n)}
  \]
  of $L'$ and then simulates $C'$ on $x'$ for $T'(N)$ steps, accepting if and
  only if $C'$ does.
  Producing $x'$ from $x$ requires $O(N \cdot \poly(n))$ time since we need only
  copy $O(n)$ bits from each block separated by the $\%$ delimiters to the next
  (namely the string $x$ and the number of the previous block).
  Using the standard simulation of cellular automata by Turing machines, the
  subsequent simulation of $C'$ requires $O(N \cdot T'(N))$ time.
  Checking whether $C'$ accepts or not can be performed in parallel to the
  simulation and requires no additional time.
  Hence, the time complexity of $S$ is
  \[
    O(N \cdot \poly(n) + N \cdot T'(N)) = O(h(n) \cdot \poly(n) \cdot T'(h(n)
    \cdot \poly(n))).
    \qedhere
  \]
\end{proof}


\section{Further Directions}
\label{sec_further_directions}

\subparagraph{$\LLT$ and two-sided error PACA.}
Besides giving a separation between one- and two-sided error,
\cref{thm_char_twosided_PACA} considerably narrows down the position of the
class of languages accepted by constant-time two-sided error PACA in the
subregular hierarchy.
Nevertheless, even though we now know the class is \enquote{sandwiched}
in-between $\LLT_{\cup\cap}$ and $\LTT$, we still do not have a precise
characterization for it.
It is challenging to tighten the inclusion from \cref{thm_char_twosided_PACA} 
because the strategy we follow relies on closure under complement, but (as we
also prove) the class of two-sided error PACA is \emph{not} closed under
complement.
It appears that clarifying the relation between said class and $\LLT_{\cup\cap}$
as well as $\LLT_{\cup\cap}$ itself and $\LLT_{\cup}$ or also $\LT$ may give
a \enquote{hint} on how to proceed.

\subparagraph{The general sublinear-time case.}
\cref{thm_PACA_eq_ACA} indicates that even polylogarithmic-time PACA can
recognize languages for which no deterministic polynomial-time algorithm is
currently known.
Although the proof of \cref{prop_PACA_inclusion_ACA} does yield explicit
examples of such languages, they are rather unsatisfactory since in order to
accept them we do not need the full capabilities of the PACA model.
(In particular, communication between blocks of cells is only required to check
certain syntactic properties of the input; once this is done, the blocks operate
independently from one another.)
It would be very interesting to identify languages where the capabilities of the
PACA model are put to more extensive use.

\subparagraph{Pseudorandom generators.}
From the opposite direction, to investigate the limitations of the PACA model,
one possibility would be to construct \emph{pseudorandom generators} (PRGs) that
fool sublinear-time PACAs.
Informally, such a PRG is a function $G\colon \binalph^{s(n)} \to
\binalph^{r(n)}$ with $s(n) \ll r(n)$ and having the property that a PACA (under
given time constraints) is incapable of distinguishing $G(x)$ from uniform when
the seed $x$ is chosen uniformly at random.
PRGs have found several applications in complexity theory (see, e.g.,
\cite{vadhan12_pseudorandomness_book} for an introduction).

\cref{thm_PACA_eq_ACA} suggests that an unconditional time-efficient
derandomization of PACAs is beyond reach of current techniques, so perhaps
\emph{space-efficient} derandomization should be considered instead.
Indeed, as a PACA can be simulated by a space-efficient machine (e.g., by
adapting the algorithm from \cite{modanese21_lower_csr}), it is possible to
recast PRGs that fool space-bounded machines (e.g.,
\cite{nisan92_pseudorandom_comb, hoza20_simple_siamjc}) as PRGs that fool PACAs.
Nevertheless, we may expect to obtain even better constructions by exploiting
the locality of PACAs (which space-bounded machines do not suffer from).


\printbibliography

@book{delorme99_cellular_book,
	doi = {10.1007/978-94-015-9153-9},
	editor = {Delorme, Marianne and Mazoyer, Jacques},
	isbn = {978-0-7923-5493-2},
	location = {Dordrecht},
	number = {460},
	publisher = {Springer Netherlands},
	series = {Mathematics and Its Applications},
	subtitle = {A Parallel Model},
	title = {Cellular Automata},
	year = {1999}
}

@incollection{kutrib09_cellular_ecss,
	author = {Kutrib, Martin},
	bibsource = {dblp computer science bibliography, https://dblp.org},
	booktitle = {Encyclopedia of Complexity and Systems Science},
	doi = {10.1007/978-0-387-30440-3\_54},
	pages = {800–823},
	title = {Cellular Automata and Language Theory},
	year = {2009}
}

@article{sommerhalder83_parallel_ai,
	author = {Sommerhalder, Rudolph and van Westrhenen, S. Christian},
	bibsource = {dblp computer science bibliography, https://dblp.org},
	doi = {10.1007/BF00290736},
	journal = {Acta Inf.},
	pages = {397–407},
	title = {Parallel Language Recognition in Constant Time by Cellular Automata},
	volume = {19},
	year = {1983}
}

@article{ibarra85_fast_tcs,
	author = {Ibarra, Oscar H. and Palis, Michael A. and Kim, Sam M.},
	bibsource = {dblp computer science bibliography, https://dblp.org},
	doi = {10.1016/0304-3975(85)90073-8},
	journal = {Theor. Comput. Sci.},
	pages = {231–246},
	title = {Fast Parallel Language Recognition by Cellular Automata},
	volume = {41},
	year = {1985}
}

@article{kim90_characterization_pd,
	acmid = {96264},
	address = {Amsterdam, The Netherlands, The Netherlands},
	author = {Kim, Sam and McCloskey, Robert},
	doi = {10.1016/0167-2789(90)90198-X},
	issn = {0167-2789},
	issue_date = {Sep. 1990},
	journal = {Phys. D},
	number = {1-3},
	numpages = {16},
	pages = {404–419},
	publisher = {Elsevier Science Publishers B. V.},
	title = {A Characterization of Constant-Time Cellular Automata Computation},
	volume = {45},
	year = {1990}
}

@incollection{terrier12_language_hnc,
	author = {Terrier, Véronique},
	bibsource = {dblp computer science bibliography, https://dblp.org},
	booktitle = {Handbook of Natural Computing},
	doi = {10.1007/978-3-540-92910-9\_4},
	pages = {123–158},
	title = {Language Recognition by Cellular Automata},
	year = {2012}
}

@book{mcnaughton71_counter-free_book,
	author = {McNaughton, Robert and Papert, Seymour},
	isbn = {0262130769},
	location = {Cambridge, MA},
	publisher = {The MIT Press},
	title = {Counter-Free Automata},
	year = {1971}
}

@inproceedings{beauquier89_factors_icalp,
	author = {Beauquier, Danièle and Pin, Jean-Eric},
	bibsource = {dblp computer science bibliography, https://dblp.org},
	booktitle = {Automata, Languages and Programming, 16th International Colloquium, ICALP89, Stresa, Italy, July 11-15, 1989, Proceedings},
	doi = {10.1007/BFb0035752},
	pages = {63–79},
	title = {Factors of Words},
	year = {1989}
}

@inproceedings{ruiz98_locally_icgi,
	author = {Ruiz, José and Boquera, Salvador España and García, Pedro},
	bibsource = {dblp computer science bibliography, https://dblp.org},
	booktitle = {Grammatical Inference, 4th International Colloquium, ICGI-98, Ames, Iowa, USA, July 12-14, 1998, Proceedings},
	doi = {10.1007/BFb0054072},
	pages = {150–161},
	title = {Locally Threshold Testable Languages in Strict Sense: Application to the Inference Problem},
	year = {1998}
}

@book{arora09_computational_book,
	author = {Arora, Sanjeev and Barak, Boaz},
	bibsource = {dblp computer science bibliography, https://dblp.org},
	isbn = {978-0-521-42426-4},
	location = {Cambridge},
	publisher = {Cambridge University Press},
	title = {Computational Complexity: {A} Modern Approach},
	year = {2009}
}

@inproceedings{yao89_circuits_stoc,
	author = {Yao, Andrew Chi-Chih},
	bibsource = {dblp computer science bibliography, https://dblp.org},
	booktitle = {Proceedings of the 21st Annual {ACM} Symposium on Theory of Computing, May 14-17, 1989, Seattle, Washington, {USA}},
	doi = {10.1145/73007.73025},
	editor = {Johnson, David S.},
	pages = {186–196},
	publisher = {{ACM}},
	title = {Circuits and Local Computation},
	year = {1989}
}

@inproceedings{garcia03_threshold_tcm,
	author = {García, Pedro and Ruiz, José},
	bibsource = {dblp computer science bibliography, https://dblp.org},
	booktitle = {Grammars and Automata for String Processing: From Mathematics and Computer Science to Biology, and Back: Essays in Honour of Gheorghe Paun},
	editor = {Martín-Vide, Carlos and Mitrana, Victor},
	pages = {243–252},
	publisher = {Taylor and Francis},
	series = {Topics in Computer Mathematics},
	title = {Threshold Locally Testable Languages in Strict Sense},
	volume = {9},
	year = {2003}
}

@book{goldreich08_computational_book,
	author = {Goldreich, Oded},
	isbn = {978-0-521-88473-0},
	location = {Cambridge},
	publisher = {Cambridge University Press},
	title = {Computational Complexity: {A} Conceptional Perspective},
	year = {2008}
}

@article{modanese21_sublinear-time_ijfcs,
	author = {Modanese, Augusto},
	doi = {10.1142/S0129054121420053},
	journal = {Int. J. Found. Comput. Sci.},
	number = {6},
	pages = {713–731},
	title = {Sublinear-Time Language Recognition and Decision by One-Dimensional Cellular Automata},
	volume = {32},
	year = {2021}
}

@article{vadhan12_pseudorandomness_book,
	author = {Vadhan, Salil P.},
	bibsource = {dblp computer science bibliography, https://dblp.org},
	doi = {10.1561/0400000010},
	journal = {Found. Trends Theor. Comput. Sci.},
	number = {1-3},
	pages = {1–336},
	title = {Pseudorandomness},
	volume = {7},
	year = {2012}
}

@article{hoza20_simple_siamjc,
	author = {Hoza, William M. and Zuckerman, David},
	bibsource = {dblp computer science bibliography, https://dblp.org},
	doi = {10.1137/19M1268707},
	journal = {{SIAM} J. Comput.},
	number = {4},
	pages = {811–820},
	title = {Simple Optimal Hitting Sets for Small-Success {RL}},
	volume = {49},
	year = {2020}
}

@article{nisan92_pseudorandom_comb,
	author = {Nisan, Noam},
	bibsource = {dblp computer science bibliography, https://dblp.org},
	doi = {10.1007/BF01305237},
	journal = {Comb.},
	number = {4},
	pages = {449–461},
	title = {Pseudorandom generators for space-bounded computation},
	url = {https://doi.org/10.1007/BF01305237},
	volume = {12},
	year = {1992}
}

@inproceedings{modanese21_lower_csr,
	author = {Modanese, Augusto},
	bibsource = {dblp computer science bibliography, https://dblp.org},
	booktitle = {Computer Science - Theory and Applications - 16th International Computer Science Symposium in Russia, {CSR} 2021, Sochi, Russia, June 28 - July 2, 2021, Proceedings},
	doi = {10.1007/978-3-030-79416-3\_18},
	editor = {Santhanam, Rahul and Musatov, Daniil},
	pages = {296–320},
	publisher = {Springer},
	series = {Lecture Notes in Computer Science},
	title = {Lower Bounds and Hardness Magnification for Sublinear-Time Shrinking Cellular Automata},
	volume = {12730},
	year = {2021}
}

@article{arrighi13_stochastic_fi,
	author = {Arrighi, Pablo and Schabanel, Nicolas and Theyssier, Guillaume},
	bibsource = {dblp computer science bibliography, https://dblp.org},
	doi = {10.3233/FI-2013-875},
	journal = {Fundam. Informaticae},
	number = {2-3},
	pages = {121–156},
	title = {Stochastic Cellular Automata: Correlations, Decidability and Simulations},
	volume = {126},
	year = {2013}
}

@article{mairesse14_around_tcs,
	author = {Mairesse, Jean and Marcovici, Irène},
	bibsource = {dblp computer science bibliography, https://dblp.org},
	doi = {10.1016/j.tcs.2014.09.009},
	journal = {Theor. Comput. Sci.},
	pages = {42–72},
	title = {Around probabilistic cellular automata},
	volume = {559},
	year = {2014}
}

@article{beauquier91_languages_tcs,
	author = {Beauquier, Danièle and Pin, Jean{-}Eric},
	bibsource = {dblp computer science bibliography, https://dblp.org},
	doi = {10.1016/0304-3975(91)90258-4},
	journal = {Theor. Comput. Sci.},
	number = {1},
	pages = {3–21},
	title = {Languages and Scanners},
	url = {https://doi.org/10.1016/0304-3975(91)90258-4; https://dblp.org/rec/journals/tcs/BeauquierP91.bib},
	volume = {84},
	year = {1991}
}

@article{suomela13_survey_acmcs,
	author = {Suomela, Jukka},
	bibsource = {dblp computer science bibliography, https://dblp.org},
	doi = {10.1145/2431211.2431223},
	journal = {{ACM} Comput. Surv.},
	number = {2},
	pages = {24:1–24:40},
	title = {Survey of local algorithms},
	volume = {45},
	year = {2013}
}


\appendix

\section{An Example for PACA Being More Efficient than DACA}
\label{appx_example_daca_vs_paca}

\begin{example}
\label{ex_daca_vs_paca}
Let $\Sigma = \{ 0,1,2,3 \}$ and consider the language
\[
  L = \{ 0^k1^l2^m3^n \mid \text{$k,l,m,n \in \N_0$ and
    (($l \ge 2$ and $m \ge 3$) or ($l \ge 3$ and $m \ge 2$))} \}.
\]
\begin{figure}
  \centering
  \begin{tikzpicture}[font=\large]
  \draw (0,0) grid (15,1);
  \begin{scope}[shift={(0.5,0.5)}]
    \node at (-2,0) {$x \in L$};
    \foreach \i in {0,...,4} \node at (\i,0) {$0$};
    \foreach \i in {5,6,7} \node at (\i,0) {$1$};
    \foreach \i in {8,9} \node at (\i,0) {$2$};
    \foreach \i in {10,...,14} \node at (\i,0) {$3$};
  \end{scope}

  \draw[shift={(0.5,0)}] (0,-1) grid (14,-2);
  \begin{scope}[shift={(1,-1.5)}]
    \node at (-2.5,0) {$z \notin L$};
    \foreach \i in {0,...,4} \node at (\i,0) {$0$};
    \foreach \i in {5,6} \node at (\i,0) {$1$};
    \foreach \i in {7,8} \node at (\i,0) {$2$};
    \foreach \i in {9,...,13} \node at (\i,0) {$3$};
  \end{scope}

  \draw (0,-3) grid (15,-4);
  \begin{scope}[shift={(0.5,-3.5)}]
    \node at (-2,0) {$y \in L$};
    \foreach \i in {0,...,4} \node at (\i,0) {$0$};
    \foreach \i in {5,6} \node at (\i,0) {$1$};
    \foreach \i in {7,8,9} \node at (\i,0) {$2$};
    \foreach \i in {10,...,14} \node at (\i,0) {$3$};
  \end{scope}

  \draw[dashed] (6,0) -- (6.2,-0.2) -- (14.8,-0.2) --(15,0);
  \draw[dashed] (5.5,-1) -- (5.7,-0.8) -- (14.3,-0.8) --(14.5,-1);
  \draw[dashed] (0,-3) -- (0.2,-2.8) -- (8.8,-2.8) --(9,-3);
  \draw[dashed] (0.5,-2) -- (0.7,-2.2) -- (9.3,-2.2) --(9.5,-2);

  \node at (4.75,-2.5) {$=$};
  \node at (10.25,-0.5) {$=$};
\end{tikzpicture}
  \caption{Comparing the words $x = 0^51^32^23^5 \in L$, $y = 0^51^22^33^5 \in
    L$, and $z = 0^51^22^23^5 \notin L$, we notice that every infix of length
    $5$ of $z$ appears in either $x$ or $y$.
    This implies there is no DACA that accepts $L$ with time complexity $3$ or
    less.}
  \label{fig_ex_xyz}
\end{figure}
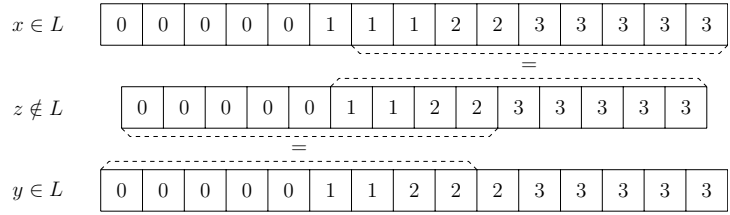%
No DACA $C$ accepts $L$ in at most $3$ steps.
This can be shown using methods from \cite{modanese21_sublinear-time_ijfcs,
sommerhalder83_parallel_ai, kim90_characterization_pd}.
Given a DACA $C$ with time complexity $3$, we can determine if $C$ accepts a
word $x \in \Sigma^\ast$ by looking only at the infixes of length $5$ (and the
prefix and suffix of length $4$) of $x$.
Consider the words $x = 0^51^32^23^5 \in L$, $y = 0^51^22^33^5 \in L$, and $z =
0^51^22^23^5 \notin L$ (\cref{fig_ex_xyz}).
Then every infix of length $5$ (and the prefix and suffix of length $4$) of $z$
appears in $x$ except for the infixes $00112$ and $01122$, which both appear in
$y$.
It follows that, if $x,y \in L(C)$, then $C$ must also accept $z$, which proves
there is no DACA for $L$ with time complexity (at most) $3$.

Nevertheless, there is a $3$-time one-sided $7/8$-error PACA $C'$ for $L$.
Checking that the input $x = 0^k1^l2^m3^n$ is such that ($x$ is of the form
$0^\ast1^\ast2^\ast3^\ast$ and) $l,m \ge 2$ can be done without need of
randomness simply by looking at the infixes of length $5$ of $x$:
Every cell collects the infix $m$ that corresponds to its position in the input
and rejects if $m$ is disallowed.
(We refer to \cite{modanese21_sublinear-time_ijfcs, sommerhalder83_parallel_ai,
kim90_characterization_pd} for the general method.)
This procedure is carried out in parallel to the one we describe next (and a
cell turns accepting if and only if it both procedures dictate it to do so).

Now we use randomness to check that one of $m \ge 3$ and $l \ge 3$ holds.
In time step $1$, every cell exposes its coin toss of step $0$ so that its
neighbors can read it and use it to choose their state in step $2$.
Let $c_\sigma$ denote the leftmost cell in which $\sigma \in \Sigma$ appears,
and let $l_\sigma$ and $r_\sigma$ be the coin tosses of the left and right
neighbors of $c_\sigma$, respectively.
We have $c_1$ accept if and only if $r_1 = 1$, $c_3$ if and only if $l_3 = 1$,
and $c_2$ if and only if $l_2 + r_2 < 2$.
All other cells accept regardless of the coin tosses they see (as long as $x$
satisfies the conditions we specified above).

For $i \in \{1,2,3\}$, let $A_i$ denote the event of cell $c_i$ accepting.
The above results in the following behavior:
If $l = m = 2$, we have $r_1 = l_2$ and $r_2 = l_3$ since the coin tosses belong
to the same cells, in which case $C'$ never accepts.
If $l \ge 3$ and $r_2$ and $l_3$ belong to the same cell (i.e., $r_2 = l_3$),
then $r_1$ and $l_2$ do not belong to the same cell and we have
\[
  \Pr[C(x,U_{T \times \abs{x}}) = 1]
  = \Pr[A_1] \Pr[A_2 \land A_3]
  = \Pr[r_1 = 1] \Pr[l_2 = 0 \land r_2 = l_3 = 1]
  = \frac{1}{8}.
\]
The case $m \ge 3$ and $r_1$ and $l_2$ belonging to the same cell is similar.
Finally, if $l \ge 3$ and $m \ge 3$, the values $r_1$, $l_2$, $r_2$, and $l_3$
are all independent and we have
\[
  \Pr[C(x,U_{T \times \abs{x}}) = 1]
  = \prod_{i=1}^3 \Pr[A_i]
  = \Pr[r_1 = 1] \Pr[l_2 + r_2 < 2] \Pr[l_3 = 1]
  > \frac{1}{8}.
\]
\end{example}

\end{document}